\documentclass[a4paper,10pt]{article} \errorcontextlines=100
\usepackage[utf8]{inputenc}
\usepackage[T1]{fontenc}
\usepackage[english]{babel}
\usepackage{authblk} 
\usepackage{amsmath,amssymb,amsthm}
\usepackage{enumerate}
\usepackage{mathtools} 
\usepackage[svgnames]{xcolor}
\usepackage{tikz,pgf}
\usepackage{bm} 
\usepackage[section]{placeins} 
\usepackage[pdfauthor={Sven Jäger and Martin Skutella}, pdftitle={Generalizing the Kawaguchi-Kyan bound to stochastic parallel machine scheduling}, pdfsubject={F.2.2 Nonnumerical Algorithms and Problems; G.2.1 Combinatorial Algorithms; G.1.6 Optimization}, pdfkeywords={Stochastic Scheduling; Parallel Machines; Approximation Algorithm; List Scheduling; Weighted Shortest (Expected) Processing Time Rule}]{hyperref}

\usepackage{pgfplots}
\usetikzlibrary{patterns}

\theoremstyle{plain}
\newtheorem{theorem}{Theorem}
\newtheorem{lem}{Lemma}
\newtheorem{corollary}{Corollary}

\theoremstyle{definition}

\theoremstyle{remark}

\newtheorem*{remark}{Remark}

\let\originalparagraph\paragraph
\renewcommand{\paragraph}[2][.]{\originalparagraph{#2#1}}

\title{Generalizing the Kawaguchi-Kyan bound to stochastic parallel machine scheduling}
\author{Sven Jäger\thanks{\url{jaeger@math.tu-berlin.de}}\ }
\author{Martin Skutella\thanks{\url{skutella@math.tu-berlin.de}}}
\affil{Institut für Mathematik, Technische Universität Berlin}

\newcommand{\E}{\mathrm E}
\newcommand{\Var}{\mathrm{Var}}
\newcommand{\N}{\mathbb N}
\newcommand{\Q}{\mathbb Q}
\newcommand{\cL}{\mathcal L}
\newcommand{\cS}{\mathcal S}
\newcommand{\cN}{\mathcal N}
\newcommand{\cM}{\mathcal M}

\begin{document}

\maketitle

\begin{abstract}
Minimizing the sum of weighted completion times on $m$ identical
parallel machines is one of the most important and classical
scheduling problems. For the stochastic variant where processing
times of jobs are random variables, Möhring, Schulz, and Uetz (1999)
presented the first and still best known approximation result
achieving, for arbitrarily many machines, performance
ratio~$1+\frac12(1+\Delta)$, where~$\Delta$ is an upper bound on the
squared coefficient of variation of the processing times. We prove
performance ratio~$1+\frac12(\sqrt{2}-1)(1+\Delta)$
for the same
underlying algorithm -- the Weighted Shortest Expected Processing
Time (WSEPT) rule. For the special case of deterministic scheduling
(i.e.,~$\Delta=0$), our bound matches the tight performance
ratio~$\frac12(1+\sqrt{2})$ of this algorithm (WSPT rule), derived by
Kawaguchi and Kyan in a 1986 landmark paper. We present several
further improvements for WSEPT's performance ratio, one of them
relying on a carefully refined analysis of WSPT yielding, for every
fixed number of machines~$m$, WSPT's exact performance ratio of
order~$\frac12(1+\sqrt{2})-O(1/m^2)$.
\end{abstract}

\section{Introduction}

In an archetypal machine scheduling problem, $n$ independent jobs 
have to be scheduled on $m$ identical parallel machines or processors. 
Each job~$j$ is specified by its processing time~$p_j>0$ and by
its weight $w_j>0$.  In a feasible schedule, every job~$j$
is processed for $p_j$ time units on one of the $m$ machines in an
uninterrupted fashion, and every machine can process at most one job
at a time.  The completion time of job~$j$ in some schedule~$\mathrm S$
is denoted by~$C_j^\mathrm S$. The goal is to compute a schedule
$\mathrm S$ that minimizes the total weighted completion time
$\sum_{j=1}^n w_j C_j^{\mathrm S}$. In the standard classification
scheme of Graham, Lawler, Lenstra, and Rinnooy Kan~\cite{GLLRK79},
this NP-hard scheduling problem is denoted by $P||\sum w_jC_j$.

\paragraph{Weighted Shortest Processing Time Rule}
By a well-known result of Smith~\cite{Smi56}, sequencing the jobs in
order of non-increasing ratios~$w_j/p_j$ gives an optimal
single-machine schedule. List scheduling in this order is known as
the Weighted Shortest Processing Time (WSPT) rule and can also be
applied to identical parallel machines, where it is a
$\frac12(1+\sqrt{2})$-approximation algorithm; see Kawaguchi and
Kyan~\cite{KK86}. A particularly remarkable aspect of Kawaguchi and
Kyan's work is that, in contrast to the vast majority
of approximation results, their analysis does not rely on some kind of
lower bound. Instead, they succeed in explicitly identifying a class
of worst-case instances. In particular, the performance
ratio~$\frac12(1+\sqrt{2})$ is tight: For every~$\varepsilon>0$ there
is a problem instance for which WSPT has approximation ratio at
least~$\frac12(1+\sqrt{2})-\varepsilon$. The instances achieving these
approximation ratios, however, have large numbers of machines
when~$\varepsilon$ becomes small. Schwiegelshohn~\cite{Sch11} gives a
considerably simpler version of Kawaguchi and Kyan's analysis.

\paragraph{Stochastic Scheduling}
Many real-world machine scheduling problems exhibit a certain degree
of uncertainty about the jobs' processing times. This characteristic
is captured by the theory of stochastic machine scheduling, where the
processing time of job~$j$ is no longer a given number~$p_j$ but a
random variable~$\bm p_j$. As all previous work in the area, we
always assume that these random variables are stochastically
independent. At the beginning, only the distributions of these random
variables are known. The actual processing time of a job becomes only
known upon its completion. As a consequence, the solution to a
stochastic scheduling problem is no longer a simple schedule, but a
so-called \emph{non-anticipative scheduling policy}. Precise
definitions on stochastic scheduling policies are given by
M{\"o}hring, Radermacher, and Weiss~\cite{MRW84}. Intuitively,
whenever a machine is idle at time~$t$, a non-anticipative scheduling
policy may decide to start a job of its choice based on the observed
past up to time~$t$ as well as the a priori knowledge of the jobs
processing time distributions and weights. It is, however, not
allowed to anticipate information about the future, i.e., the actual
realizations of the processing times of jobs that have not yet
finished by time~$t$.

It follows from simple examples that, in general, a non-anticipative
scheduling policy cannot yield an optimal schedule for each possible
realization of the processing times. We are therefore looking for a
policy which minimizes the objective in expectation. For the
stochastic scheduling problem considered in this paper, the
goal is to find a non-anticipative scheduling policy that minimizes
the expected total weighted completion time. This problem is denoted
by $P|\bm p_j \sim \mathrm{stoch}|\E[\sum w_j \bm C_j]$. 

\paragraph{Weighted Shortest Expected Processing Time Rule}
The stochastic analogue of the WSPT rule is greedily scheduling the
jobs in order of non-increasing ratios $w_j/\E[\bm p_j]$. Whenever a
machine is idle, the Weighted Shortest Expected Processing Time
(WSEPT) rule immediately starts the next job in this order. For a
single machine this is again optimal; see Rothkopf~\cite{Rot66}.
For identical parallel machines, Cheung, Fischer, Matuschke, and
Megow~\cite{CFMM14} and Im, Moseley, and Pruhs~\cite{IMP15}
independently show that WSEPT does not even achieve constant
performance ratio.
More precisely, for every $R>0$ there is a problem
instance for which WSEPT's expected total weighted completion time is
at least~$R$ times the expected objective function value of an
optimal non-anticipative scheduling policy. Even in the special case of
exponentially distributed processing times, Jagtenberg,
Schwiegelshohn, and Uetz~\cite{JSU13} show a lower bound of~$1.243$
on WSEPT's performance. 
On the positive side,
WSEPT is an optimal policy for the special case of unit weight jobs
with stochastically ordered processing times,
$P|\bm p_j\sim \mathrm{stoch}(\preceq_{\mathrm{st}}) |E[\sum \bm C_j]$;
see Weber, Varaiya, and Walrand~\cite{WVW86}.
Moreover,
Weiss~\cite{Wei90,Wei92} proves asymptotic optimality of WSEPT.
M{\"o}hring, Schulz, and Uetz~\cite{MSU99} show that the WSEPT rule
achieves performance ratio $1+\frac12(1+\Delta)(1-\frac1m)$,
where~$\Delta$ is an upper bound on the squared coefficient of
variation of the processing times.

\paragraph{Further Approximation Results from the Literature}
While there is a PTAS for the deterministic problem
$P||\sum w_jC_j$~\cite{SW00}, no constant-factor approximation
algorithm is known for the stochastic problem
$P|\bm p_j \sim \mathrm{stoch}|\E[\sum w_j \bm C_j]$. WSEPT's
performance ratio~$1+\frac12(1+\Delta)$ (for arbitrarily many
machines) proven by Möhring et al.~\cite{MSU99} is the best
hitherto known performance ratio. The only known
approximation ratio not depending on the jobs' squared coefficient of
variation~$\Delta$ is due to Im et al.~\cite{IMP15} who,
for the special case of unit job weights
$P|\bm p_j \sim \mathrm{stoch}|\E[\sum\bm C_j]$, present an
$O(\log^2n+m\log n)$-approximation algorithm.

The performance ratio~$1+\frac12(1+\Delta)$ has been carried over
to different generalizations of
$P|\bm p_j \sim \mathrm{stoch}|\E[\sum w_j \bm C_j]$. Megow, Uetz,
and Vredeveld~\cite{MUV06} show that it also applies if jobs arrive
online in a list and must immediately and irrevocably be assigned to
machines, on which they can be sequenced optimally. An approximation
algorithm with this performance ratio for the problem on unrelated
parallel machines is designed by Skutella,
Sviridenko, and Uetz~\cite{SSU16}.
If these two features are combined, i.e., in the online list-model
with unrelated machines, Gupta, Moseley, Uetz, and Xie~\cite{GMUX17}
develop a $(8+4\Delta)$-approximation algorithm.

The performance
ratios are usually larger if jobs are released over time: In the
offline setting, the best known approximation algorithm has
performance ratio $2+\Delta$; see Schulz~\cite{Sch08}. This performance ratio
is also achieved for unrelated machines \cite{SSU16} and by a
randomized online algorithm \cite{Sch08}. There exist furthermore 
a deterministic $(\max\{2.618, 2.309+1.309\Delta\})$-approximation
for the online setting \cite{Sch08} and a deterministic 
$(144+72\Delta)$-approximation for the online setting on unrelated
parallel machines \cite{GMUX17}.

\paragraph{Our Contribution and Outline}
We present the first progress on the\linebreak approximability of the basic
stochastic scheduling problem on identical\linebreak parallel machines with
expected total weighted completion time objective\linebreak
$P|\bm p_j \sim \mathrm{stoch}|\E[\sum w_j \bm C_j]$ since the
seminal work of M{\"o}hring et al.~\cite{MSU99}; see
Figure~\ref{fig:performance}.
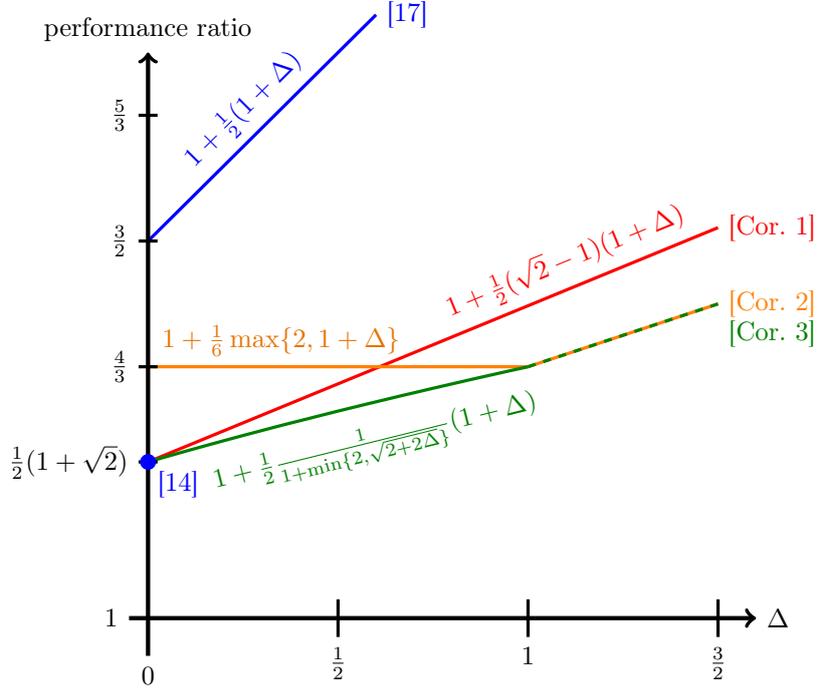
\begin{figure}[t]
\centering
\begin{tikzpicture}[x=5cm,y=10cm,domain=0:1]
\draw [blue,very thick] (0,1.5) to node [sloped,above] {$1+\tfrac12(1+\Delta)$} +(0.6,0.6/2) node [right] {\cite{MSU99}};
\draw [red,very thick] (0,1.2071) to node [sloped, above,near end] {$1+\tfrac12(\sqrt2-1)(1+\Delta)$} +(1.5,1.5*0.2071) node [right] {[Cor.~\ref{red performance guarantee}]};
\draw [orange,very thick] (0,4/3) to node [above,pos=0.35,color=orange!90!black] {$1+\tfrac16\max\{2,1+\Delta\}$} ++(1,0) -- +(0.5,0.5/6) node [right] {[Cor.~\ref{orange performance guarantee}]};
\draw [green!50!black,very thick] plot (\x,{1+1/2*1/(1+sqrt(2+2*\x))*(1+\x)});
\node [rotate=14,green!50!black] at (0.6,1.23) {$1+\tfrac 1 2 \frac{1}{1+\min\{2,\sqrt{2 + 2\Delta}\}}(1+\Delta)$};
\draw [green!50!black,very thick,dashed] (1,4/3) to +(0.5,0.5/6) node[right,yshift=-1.1em] {[Cor.~\ref{green performance guarantee}]};

\draw [->,ultra thick] (-0.05,1) node [left] {$1$} -- (1.6,1) node [right] {$\Delta$}; 
\draw [->,ultra thick] (0,0.95) node [below] {$0$} -- (0,1.75) node [above] {performance ratio};
\foreach \i/\j in {0.5/\frac12,1/1,1.5/\frac32}
\draw [very thick] (\i,0.975) node [below] {$\j$} -- +(0,0.05);
\foreach \i/\j in {1.2071/\frac12(1+\sqrt2),1.3333/\frac43,1.5/\frac32,1.6667/\frac53}
\draw [very thick] (-0.025,\i) node [left] {$\j$} -- +(0.05,0);

\draw [blue] (0,1.2071) node [below right] {\cite{KK86}} node [circle,fill,minimum size=6pt,inner sep=0pt] {};
\end{tikzpicture}
\caption{Bounds on WSEPT's performance ratio}
\label{fig:performance}
\end{figure}
We prove
that WSEPT achieves performance ratio
\begin{align}
1+\tfrac12\min\left\{\frac{\sqrt{(2m-k_m)k_m}-k_m}{m},
\frac{1}{1+\min\{2,\sqrt{2+2\Delta}\}}\right\}(1+\Delta),
\label{eq:perf_ratio}
\end{align}
where $k_m \coloneqq \bigl\lfloor\bigl(1-\frac12\sqrt 2\bigr) m \bigr\rceil$
is the nearest integer to $\bigl(1-\frac12\sqrt 2\bigr)m$.
Notice that, for every number of machines~$m$, the performance ratio
given by the first term of the minimum in~\eqref{eq:perf_ratio} is
bounded from above by~$1+\tfrac12(\sqrt{2}-1)(1+\Delta)$, and for
$m \to \infty$ it converges to this bound. As
$(1+\min\{2,\sqrt{2+2\Delta}\})^{-1} \le \sqrt 2-1$ for all
$\Delta > 0$, when considering an arbitrary number of machines, 
the second term in the minimum dominates the first term.
In the following, we list several points that emphasize the
significance of the new performance ratio~\eqref{eq:perf_ratio}.
\begin{itemize}
\item For the special case of deterministic scheduling
(i.e.,~$\Delta=0$), the machine-independent performance ratio
in~\eqref{eq:perf_ratio} matches the Kawaguchi-Kyan\linebreak
bound~$\frac12(1+\sqrt{2})$, which is known to be tight~\cite{KK86}. In
particular, we dissolve the somewhat annoying discontinuity of the
best previously known bounds~\cite{KK86,MSU99} at~$\Delta=0$; see
Figure~\ref{fig:performance}.
\item Again for deterministic jobs, our machine-dependent
bound\linebreak $1+\frac12(\sqrt{(2m-k_m)k_m}-k_m)/m$ is tight and slightly
improves the 30 years old Kawaguchi-Kyan bound for every fixed
number of machines~$m$; see Figure~\ref{graph of performance guarantee}.
\begin{figure}
 \centering
 \begin{tikzpicture}
  \begin{axis}[height=6cm,width=12cm,xmin=0,xmax=25.9,ymin=1.18,ymax=1.211,axis x line=bottom, axis y line=middle, xlabel={$m$}, ylabel={performance ratio}, x label style={at={(axis description cs:1.01,0.15)},anchor=west}, y label style={at={(axis description cs:0,1.01)},anchor=south},samples=100]
   \addplot[red,thick] coordinates {(0,{(1+sqrt(2))/2}) (25,{(1+sqrt(2))/2})};
   \pgfmathdeclarefunction{k}{1}{\pgfmathparse{floor((1-sqrt(2)/2)*#1+0.5)}}
   \addplot[domain=1.9:5.121] {1+1/2*(sqrt((2*x-k(x))*k(x))-k(x))/x};
   \addplot[domain=5.122:8.535] {1+1/2*(sqrt((2*x-k(x))*k(x))-k(x))/x};
   \addplot[domain=8.536:11.949] {1+1/2*(sqrt((2*x-k(x))*k(x))-k(x))/x};
   \addplot[domain=11.95:15.3635] {1+1/2*(sqrt((2*x-k(x))*k(x))-k(x))/x};
   \addplot[domain=15.3645:18.778] {1+1/2*(sqrt((2*x-k(x))*k(x))-k(x))/x};
   \addplot[domain=18.779:22.192] {1+1/2*(sqrt((2*x-k(x))*k(x))-k(x))/x};
   \addplot[domain=22.193:25] {1+1/2*(sqrt((2*x-k(x))*k(x))-k(x))/x};
   \addplot[only marks, magenta, mark size=1.2pt, samples at={2,...,25}] {1+1/2*(sqrt((2*x-k(x))*k(x))-k(x))/x};
  \end{axis}
 \end{tikzpicture}
 \caption{Graph of the function $m \mapsto 1+\frac12(\sqrt{(2m-k_m)k_m}-k_m)/m$, which for $m \in \N$ gives the worst-case approximation ratio of WSPT for $P||\sum w_j C_j$ with $m$ machines (dots), compared to the machine-independent Kawaguchi-Kyan bound} \label{graph of performance guarantee}
\end{figure}
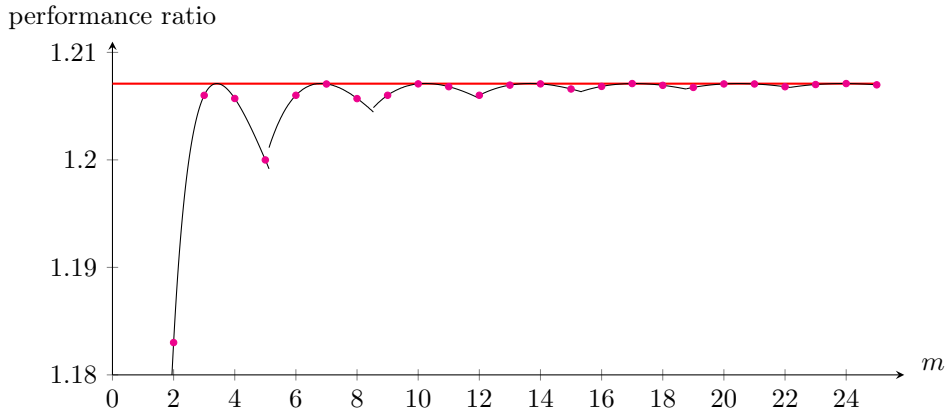
\item For exponentially distributed processing times ($\Delta=1$),
our results imply that WSEPT achieves performance ratio~$4/3$. This
solves an open problem by Jagtenberg et al.~\cite{JSU13}, who give a
lower bound of~$1.243$ on WSEPT's performance and ask for an
improvement of the previously best known upper bound of~$2-1/m$ due
to M{\"o}hring et al.~\cite{MSU99}.
\item WSEPT's performance bound due to Möhring et al.~\cite{MSU99} also 
holds for the MinIncrease policy, introduced by
Megow et al.~\cite{MUV06}, which is a fixed-assignment policy,
i.e., before the execution of the jobs starts,
it determines for each job on which machine it is going to be processed.
Our stronger bound, together with a lower bound in~\cite{SSU16}, shows
that WSEPT actually beats every fixed-assignment policy.
\end{itemize}

The improved performance ratio in~\eqref{eq:perf_ratio} is derived
as follows. In Section~\ref{Performance guarantee WSEPT} we present
one of the key results of this paper (see
Theorem~\ref{Theorem Performance guarantee WSEPT} below): If WSPT has
performance ratio~$1+\beta$ for some~$\beta$, then WSEPT achieves
performance ratio~$1+\beta(1+\Delta)$ for the stochastic scheduling
problem. For the Kawaguchi-Kyan bound~$1+\beta=\frac12(1+\sqrt{2})$,
this yields performance ratio~$1+\frac12(\sqrt{2}-1)(1+\Delta)$. It
is also interesting to notice that the performance ratio of Möhring
et al.~\cite{MSU99} follows from this theorem by plugging in
$1+\beta = 3/2 - 1/(2m)$, which is WSPT's performance ratio
obtained from the bound of Eastman, Even, and Isaacs~\cite{EEI64}.
We generalize 
Theorem~\ref{Theorem Performance guarantee WSEPT} to performance
ratios w.r.t.\ the weighted sum of\linebreak $\alpha$-points as objective
function, where the $\alpha$-point of a job~$j$ is the point in time 
when it has been processed for exactly $\alpha p_j$ time units. 

The theorems derived in Section~\ref{Performance guarantee WSEPT}
provide tools to carry over bounds for the WSPT rule
to the WSEPT rule. The concrete performance ratio for the WSEPT rule
obtained this way thus depends on good bounds for the WSPT
rule. In Section~\ref{Performance guarantees WSPT based on alpha-points}
we derive performance ratios for WSPT
w.r.t.\ the weighted sum of $\alpha$-points objective.
For $\alpha=\tfrac 1 2$ this performance ratio
follows easily from a result by Avidor, Azar, and
Sgall~\cite{AAS01} for the 2-norm of the machine load vector.
As a consequence we obtain performance ratio
$1+\tfrac 1 6\max\{2,1+\Delta\}$ for WSEPT. By optimizing the choice 
of $\alpha$, we finally obtain the performance ratio
$1+\frac 1 2 (1+\min\{2,\sqrt{2+2\Delta}\})^{-1}(1+\Delta)$.
The various bounds on WSEPT's performance 
derived
in Sections~\ref{Performance guarantee WSEPT} and
\ref{Performance guarantees WSPT based on alpha-points}
are illustrated in Figure~\ref{fig:performance}. Finally, in
Section~\ref{Performance guarantee WSPT} the analysis
of Schwiegelshohn~\cite{Sch11} for the WSPT rule
is refined for every fixed number of machines~$m$,
entailing the machine-dependent bound for the WSEPT rule in
\eqref{eq:perf_ratio}.

\section{Performance ratio of the WSEPT rule} \label{Performance guarantee WSEPT}
 
Let $\Delta \ge \Var[\bm p_j]/\E[\bm p_j]^2$ for all $j \in \{1,\dotsc,n\}$. In Theorems~\ref{Theorem Performance guarantee WSEPT} and \ref{Performance WSEPT based on alpha-points} we demonstrate how performance ratios for the WSPT rule for deterministic scheduling can be carried over to stochastic scheduling. Theorem~\ref{Theorem Performance guarantee WSEPT} starts out from a performance ratio for WSPT with respect to the usual objective function: the weighted sum of completion times. In Theorem~\ref{Performance WSEPT based on alpha-points} this is generalized insofar as a performance ratio of WSPT with respect to the weighted sum of $\alpha$-points is taken as a basis.

\begin{theorem} \label{Theorem Performance guarantee WSEPT}
 If the WSPT rule on $m$ machines has performance ratio $1+\beta_m$ for the problem $P||\sum w_j C_j$, then the WSEPT rule achieves performance ratio $1+\beta_m(1+\Delta)$ for $P|\bm p_j \sim \mathrm{stoch}|\E[\sum w_j \bm C_j]$.
\end{theorem}

The reason why the bound for WSPT does not directly carry over to WSEPT is that under a specific realization of the processing times the schedule obtained by the WSEPT policy may differ from the WSPT schedule for this realization. Still, under every realization the WSEPT schedule is a list schedule. Hence, usually a bound that is valid for every list schedule is used: The objective value of a list schedule on $m$ machines is at most $1/m$ times the objective value of the list schedule on a single machine plus $(m-1)/m$ times the weighted sum of processing times. This bound, holding because a list schedule assigns each job to the currently least loaded machine, is applied realizationwise to obtain a corresponding bound on the expected values in stochastic scheduling (cf.\ \cite[Lemma~4.1]{MSU99}), which can be compared to an LP-based lower bound on the expected total completion time under an optimal policy. 

In order to benefit from the precise bounds known for the WSPT rule nevertheless, we regard the following modified stochastic scheduling problem: For each job, instead of its weight $w_j$, we are given a weight factor $\rho_j$. The actual weight of a job is $\rho_j$ times its actual processing time, i.e., if a job takes longer, it also becomes more important. The goal is again to minimize the total weighted completion time. For the thus defined stochastic scheduling problem list scheduling in order of $\rho$-values has the nice property that it creates a WSPT schedule in every realization. So, for this scheduling problem any performance ratio of WSPT directly carries over to this list scheduling policy. In the following proof of Theorem~\ref{Theorem Performance guarantee WSEPT} we first compare the expected total weighted completion time of a WSEPT schedule for the original problem to the expected objective value list scheduling in order of $\rho_j$ for the modified problem, then apply the performance ratio of the WSPT rule, and finally compare the expected total weighted completion time of the optimal schedule for the modified problem to the expected objective value of the optimal policy for the original problem. The transitions between the two problems lead to the additional factor of $1+\Delta$ in the performance ratio.

\begin{proof}
Consider an instance of $P|\bm p_j \sim \mathrm{stoch}|\E[\sum w_j \bm C_j]$ consisting of $n$ jobs and $m$ machines, and let $\beta \coloneqq \beta_m$ and $\rho_j \coloneqq w_j/\E[\bm p_j]$
for $j \in \{1,\dotsc,n\}$. For every realization $(p_1,\dotsc,p_n)$ of the processing times we consider the instance $I(p_1,\dotsc,p_n)$ of $P||\sum w_j C_j$ which consists of $n$ jobs with processing times $p_1,\dotsc,p_n$ and weights $\rho_1 p_1, \dotsc, \rho_n p_n$, so that the jobs in the instance $I(p_1,\dotsc,p_n)$ have Smith ratios $\rho_1,\dotsc,\rho_n$ for all possible realizations. Therefore, for every realization $(p_1,\dotsc,p_n)$ the schedule obtained by the WSEPT policy is a WSPT schedule for $I(p_1,\dotsc,p_n)$. Let $C_j^{\mathrm{WSEPT}}(p_1,\dotsc,p_n)$ denote the completion time of job~$j$ in the schedule obtained by the WSEPT policy under the realization $(p_1,\dotsc,p_n)$, let $C_j^*(I(p_1,\dotsc,p_n))$ denote its completion time in an optimal schedule for $I(p_1,\dotsc,p_n)$, and let $C_j^{\Pi^*}(p_1,\dotsc,p_n)$ denote $j$'s completion time in the schedule constructed by an optimal stochastic scheduling policy under the realization $(p_1,\dotsc,p_n)$. For every realization $(p_1,\dotsc,p_n)$ of the processing times, since the WSEPT schedule follows the WSPT rule for $I(p_1,\dotsc,p_n)$, its objective value is bounded by
\[\sum_{j=1}^n (\rho_j p_j) C_j^{\mathrm{WSEPT}}(p_1,\dotsc,p_n) \le (1+\beta) \cdot \sum_{j=1}^n (\rho_j p_j) C_j^*(I(p_1,\dotsc,p_n)).\]
On the other hand, as the schedule obtained by an optimal stochastic scheduling policy is feasible for the instance $I(p_1,\dotsc,p_n)$, it holds that
\[\sum_{j=1}^n (\rho_j p_j) C_j^*(I(p_1,\dotsc,p_n)) \le \sum_{j=1}^n (\rho_j p_j) C_j^{\Pi^*}(p_1,\dotsc,p_n).\]
 Putting these two inequalities together and taking expectations, we get the inequality
 \[\E\left[\sum_{j=1}^n \rho_j \bm p_j \bm C_j^{\mathrm{WSEPT}}\right] \le (1+\beta) \cdot \E\left[ \sum_{j=1}^n \rho_j \bm p_j \bm C_j^{\Pi^*}\right],\]
 where $\bm C_j^{\mathrm{WSEPT}} = C_j^{\mathrm{WSEPT}}(\bm p_1,\dotsc,\bm p_n)$ and $\bm C_j^{\Pi^*} = C_j^{\Pi^*}(\bm p_1,\dotsc,\bm p_n)$.
 Using the latter inequality, we can bound the expected total weighted completion time of the WSEPT rule:
 \begin{align*}
  &\E\left[\sum_{j=1}^n w_j \bm C_j^{\mathrm{WSEPT}}\right] 
  = \sum_{j=1}^n \rho_j \E[\bm p_j] \E[\bm C_j^{\mathrm{WSEPT}}] \\
  \stackrel{(*)}=\ &\sum_{j=1}^n \rho_j \E[\bm p_j \bm C_j^{\mathrm{WSEPT}}] - \sum_{j=1}^n \rho_j \Var[\bm p_j] \\
  =\ &\E\left[\sum_{j=1}^n \rho_j \bm p_j \bm C_j^{\mathrm{WSEPT}}\right] - \sum_{j=1}^n \rho_j \Var[\bm p_j] \\
  \le\ &(1+\beta) \cdot \E\left[\sum_{j=1}^n \rho_j \bm p_j \bm C_j^{\Pi^*}\right] - \sum_{j=1}^n \rho_j \Var[\bm p_j] \\
  =\ &(1+\beta) \cdot \sum_{j=1}^n \rho_j \E[\bm p_j \bm C_j^{\Pi^*}] - \sum_{j=1}^n \rho_j \Var[\bm p_j]\\
  \stackrel{(*)}=\ &(1+\beta) \cdot \left(\sum_{j=1}^n \rho_j \E[\bm p_j] \E[\bm C_j^{\Pi^*}] + \sum_{j=1}^n \rho_j \Var[\bm p_j]\right) - \sum_{j=1}^n \rho_j \Var[\bm p_j] \\
  =\ &(1+\beta) \cdot \sum_{j=1}^n w_j \E[\bm C_j^{\Pi^*}] + \beta \sum_{j=1}^n \rho_j \Var[\bm p_j] \\
  \le\ &(1+\beta) \cdot \sum_{j=1}^n w_j \E[\bm C_j^{\Pi^*}] + \Delta\beta \sum_{j=1}^n w_j \E[\bm p_j] \\
  \le\ &(1+\beta(1+\Delta)) \cdot \sum_{j=1}^n w_j \E[\bm C_j^{\Pi^*}] = (1+\beta(1+\Delta)) \cdot \E\left[\sum_{j=1}^n w_j \bm C_j^{\Pi^*}\right].
 \end{align*}
 The equalities marked with $(*)$ hold because for any stochastic scheduling policy~$\Pi$ and all~$j$
 \begin{multline*}
  \E[\bm p_j \bm C_j^{\Pi}] = \E[\bm p_j (\bm S_j^{\Pi} + \bm p_j)] 
  = \E[\bm p_j \bm S_j^{\Pi}] + \E[\bm p_j^2] \\= \E[\bm p_j] \E[\bm S_j^{\Pi}] + \E[\bm p_j]^2 + \Var[\bm p_j] = \E[\bm p_j] \E[\bm C_j^{\Pi}] + \Var[\bm p_j],
 \end{multline*}
 where $\bm S_j^\Pi$ denotes the starting time of job $j$ under policy $\Pi$. The independence of $\bm p_j$ and $\bm S_j^\Pi$ follows from the non-anticipativity of policy~$\Pi$, and the last inequality uses the fact that $\E[\bm p_j] \le \E[\bm C_j^{\Pi^*}]$ for every job~$j$.
\end{proof}

By plugging in the Kawaguchi-Kyan bound, we immediately get the following performance ratio (see Figure~\ref{fig:performance}).

\begin{corollary} \label{red performance guarantee}
 The WSEPT rule has performance ratio $1 + \frac 1 2 (\sqrt 2 - 1)(1+\Delta)$ for the problem $P|\bm p_j \sim \mathrm{stoch}|\E[\sum w_j \bm C_j]$.
\end{corollary}

For $\alpha \in [0,1]$ the \emph{$\alpha$-point} $C_j^{\mathrm S}(\alpha)$ of a job~$j$ is the point in time at which it has been processed for $\alpha p_j$ time units. Introduced by Hall, Shmoys, and Wein~\cite{HSW96} in order to convert a preemptive schedule into a non-preemptive one, the concept of $\alpha$-points is often used in the \emph{design} of algorithms (see e.g.\ \cite{Goe97,CMNS01,GQS+02,Sku16}). In contrast, we use them to define an alternative objective function in order to improve the \emph{analysis} of the WSEPT rule.

We consider as objective function the weighted sum of $\alpha$-points $\sum_{j=1}^n w_j C_j^{\mathrm S}(\alpha)$ for $\alpha \in [0,1]$, generalizing the weighted sum of completion times. For every $\alpha$ the weighted sum of $\alpha$-points differs only by the constant $(1-\alpha)\sum_{j=1}^n w_j p_j$ from the weighted sum of completion times. So as for optimal solutions the objective functions are equivalent. The same applies to the stochastic variant: Here the two objectives differ by the constant $(1-\alpha)\sum_{j=1}^n w_j \E[\bm p_j]$, whence they have the same optimal policies. We now generalize Theorem~\ref{Theorem Performance guarantee WSEPT} to the (expected) weighted sum of $\alpha$-points.

\begin{theorem} \label{Performance WSEPT based on alpha-points}
 If the WSPT rule has performance ratio $1+\beta$ for the problem $P||\sum w_j C_j(\alpha)$, then the WSEPT rule has performance ratio $1+\beta(1+\Delta)$ for the problem $P|\bm p_j \sim \mathrm{stoch}|\E[\sum w_j\bm C_j(\alpha)]$ and $1 + \beta \cdot \max\{1,\alpha(1+\Delta)\}$ for the problem $P|\bm p_j \sim \mathrm{stoch} |\E[\sum w_j \bm C_j]$.
\end{theorem}

The proof uses the same idea as the proof of Theorem~\ref{Theorem Performance guarantee WSEPT}: The bound for $P||\sum w_j C_j(\alpha)$ is again applied realizationwise to the modified stochastic problem described above.

\begin{proof}
 We use the notation of the proof of Theorem~\ref{Theorem Performance guarantee WSEPT} for $\alpha$-points $C_j(\alpha)$ instead of completion times $C_j$. Then we have by assumption that
 \begin{align*} 
  \sum_{j=1}^n (\rho_j p_j) C_j^{\mathrm{WSEPT}}(\alpha)(p_1,\dotsc,p_n) &\le (1+\beta) \sum_{j=1}^n (\rho_j p_j) C_j^*(\alpha)(I(p_1,\dotsc,p_n))\\
  &\le (1+\beta) \sum_{j=1}^n (\rho_jp_j) C_j^{\Pi^*}(\alpha)(p_1,\dotsc,p_n),
 \end{align*} 
 where the last inequality holds because the schedule obtained by the optimal stochastic scheduling policy is feasible for $I(p_1,\dotsc,p_n)$.
 This carries over to expected values:
 \begin{equation}
  \E\left[\sum_{j=1}^n \rho_j \bm p_j \bm C_j^{\mathrm{WSEPT}}(\alpha)\right] \le (1+ \beta)\cdot \E\left[\sum_{j=1}^n \rho_j \bm p_j \bm C_j^{\Pi^*}(\alpha)\right]. \label{eq:inequality expected values random weights}
 \end{equation}
 Furthermore, for any non-anticipative stochastic scheduling policy
 \begin{align}
  \E[\bm p_j \bm C_j^\Pi(\alpha)] &= 
   \E[\bm p_j \bm S_j^\Pi] + \alpha \E[\bm p_j^2] = \E[\bm p_j] \E[\bm S_j^\Pi] + \alpha(\E[\bm p_j]^2 + \Var[\bm p_j])\nonumber \\
  &= \E[\bm p_j] \E[\bm C_j^\Pi(\alpha)] + \alpha \Var[\bm p_j] \label{eq:weighted expected alpha-point} \\
  &= \E[\bm p_j] \E[\bm C_j^\Pi] + \alpha \Var[\bm p_j] - (1-\alpha) \E[\bm p_j]^2. \label{eq:weighted expected completion time}
 \end{align}

Therefore,
 \begin{align*} 
  &\E\left[\sum_{j=1}^n w_j \bm C_j^{\mathrm{WSEPT}}(\alpha)\right] = \sum_{j=1}^n \rho_j \E[\bm p_j] \E[\bm C_j^{\mathrm{WSEPT}}(\alpha)]\\
  \stackrel{\eqref{eq:weighted expected alpha-point}}=\ &\sum_{j=1}^n \rho_j \E[\bm p_j \bm C_j^{\mathrm{WSEPT}}(\alpha)] - \alpha \sum_{j=1}^n \rho_j \Var[\bm p_j]) \\
  =\ &\E\left[\sum_{j=1}^n \rho_j \bm p_j \bm C_j^{\mathrm{WSEPT}}(\alpha)\right] - \alpha \sum_{j=1}^n \rho_j \Var[\bm p_j] \\
  \stackrel{\eqref{eq:inequality expected values random weights}}\le\ &(1+\beta) \cdot \E\left[\sum_{j=1}^n \rho_j \bm p_j \bm C_j^{\Pi^*}(\alpha)\right] - \alpha \sum_{j=1}^n \rho_j \Var[\bm p_j] \\
  =\ &(1+\beta) \cdot \sum_{j=1}^n \rho_j \E[\bm p_j \bm C_j^{\Pi^*}(\alpha)] - \alpha \sum_{j=1}^n \rho_j \Var[\bm p_j] \\
  \stackrel{\eqref{eq:weighted expected alpha-point}}=\ &(1+\beta) \cdot \left(\sum_{j=1}^n \rho_j \E[\bm p_j] \E[\bm C_j^{\Pi^*}(\alpha)] + \alpha \sum_{j=1}^n \rho_j \Var[\bm p_j]\right) - \alpha \sum_{j=1}^n \rho_j \Var[\bm p_j] \\
  =\ &(1+\beta) \cdot \sum_{j=1}^n w_j \E[\bm C_j^{\Pi^*}(\alpha)] + \alpha \beta \sum_{j=1}^n \rho_j \Var[\bm p_j] \\
  \le\ &(1+\beta) \cdot \sum_{j=1}^n w_j \E[\bm C_j^{\Pi^*}(\alpha)] + \alpha \beta \Delta \sum_{j=1}^n w_j \underbrace{\E[\bm p_j]}_{\le \frac{\E[\bm C_j^{\Pi^*}(\alpha)]}{\alpha}} \\
  \le\ &(1+\beta(1+\Delta)) \cdot \sum_{j=1}^n w_j \E[\bm C_j^{\Pi^*}(\alpha)] = (1+\beta(1+\Delta)) \cdot \E\left[\sum_{j=1}^n w_j \bm C_j^{\Pi^*}(\alpha)\right],
 \end{align*}
 and
 \begin{align*}
  &\E\left[\sum_{j=1}^n w_j \bm C_j^{\mathrm{WSEPT}}\right] = \sum_{j=1}^n \rho_j \E[\bm p_j] \E[\bm C_j^{\mathrm{WSEPT}}]\\
  \stackrel{\eqref{eq:weighted expected completion time}}=\ &\sum_{j=1}^n \rho_j \E[\bm p_j \bm C_j^{\mathrm{WSEPT}}(\alpha)] - \sum_{j=1}^n \rho_j (\alpha \Var[\bm p_j] - (1-\alpha)\E[\bm p_j]^2) \\
  =\ &\E\left[\sum_{j=1}^n \rho_j \bm p_j \bm C_j^{\mathrm{WSEPT}}(\alpha)\right] - \sum_{j=1}^n \rho_j (\alpha \Var[\bm p_j] - (1-\alpha)\E[\bm p_j]^2) \\
  \stackrel{\eqref{eq:inequality expected values random weights}}\le\ &(1+\beta) \cdot \E\left[\sum_{j=1}^n \rho_j \bm p_j \bm C_j^{\Pi^*}(\alpha)\right] - \sum_{j=1}^n \rho_j (\alpha \Var[\bm p_j] - (1-\alpha)\E[\bm p_j]^2) \\
  =\ &(1+\beta) \cdot \sum_{j=1}^n \rho_j \E[\bm p_j \bm C_j^{\Pi^*}(\alpha)] - \sum_{j=1}^n \rho_j (\alpha \Var[\bm p_j] - (1-\alpha)\E[\bm p_j]^2) \\
  \stackrel{\eqref{eq:weighted expected completion time}}=\ &(1+\beta) \cdot \left(\sum_{j=1}^n \rho_j \E[\bm p_j] \E[\bm C_j^{\Pi^*}] + \sum_{j=1}^n \rho_j (\alpha \Var[\bm p_j] - (1-\alpha) \E[\bm p_j]^2)\right) \\
     &\hspace{4em} - \sum_{j=1}^n \rho_j (\alpha \Var[\bm p_j] - (1-\alpha)\E[\bm p_j]^2) \\
  =\ &(1+\beta) \cdot \sum_{j=1}^n w_j \E[\bm C_j^{\Pi^*}] + \beta \sum_{j=1}^n \rho_j (\alpha \Var[\bm p_j] - (1-\alpha) \E[\bm p_j]^2) \\
  \le\ &(1+\beta) \cdot \sum_{j=1}^n w_j \E[\bm C_j^{\Pi^*}] + \beta \sum_{j=1}^n \rho_j (\alpha \Delta - (1-\alpha)) \E[\bm p_j]^2 \\
  =\ &(1+\beta) \cdot \sum_{j=1}^n w_j \E[\bm C_j^{\Pi^*}] + \beta (\alpha(1+\Delta)-1) \sum_{j=1}^n w_j \underbrace{\E[\bm p_j]}_{\in [0,\E[\bm C_j^{\Pi^*}]]} \\
  \le\ &(1+\beta) \cdot \sum_{j=1}^n w_j \E[\bm C_j^{\Pi^*}] + \beta \cdot \max\{0,\alpha (1+\Delta) - 1\} \sum_{j=1}^n w_j \E[\bm C_j^{\Pi^*}] \\
  \le\ &(1 + \beta \cdot \max\{1,\alpha(1+\Delta)\}) \cdot \E\left[\sum_{j=1}^n w_j \bm C_j^{\Pi^*}\right]. \qedhere
\end{align*}
\end{proof}

Theorem~\ref{Theorem Performance guarantee WSEPT} follows from Theorem~\ref{Performance WSEPT based on alpha-points} by plugging in $\alpha = 1$.

\section{Performance ratios for WSPT with weighted sum of $\alpha$-points objective} \label{Performance guarantees WSPT based on alpha-points}

In this section we derive performance ratios for $P||\sum w_j C_j(\alpha)$. The two classical performance ratios for $P||\sum w_j C_j$ of Eastman, Even, and Isaacs~\cite{EEI64} and of Kawaguchi, and Kyan~\cite{KK86} can both be generalized to this problem. The Eastman-Even-Isaacs bound can be generalized for every $\alpha \in (0,1]$, whereas the Kawaguchi-Kyan bound carries over only for $\alpha \in [\tfrac 1 2, 1]$. In return, the generalized Kawaguchi-Kyan bound is better for these $\alpha$.

For a problem instance $I$ denote by~$\cN(I)$ the job set, by $C_j^{\mathrm{WSPT}}(\alpha)(I)$ the $\alpha$-point of job~$j$ in the WSPT schedule for $I$, and by $C^*_j(\alpha)(I)$ the $\alpha$-point of job~$j$ in some fixed (`the') optimal schedule for $I$. Hence $C_j^{\mathrm{WSPT}}(1)(I) = C_j^{\mathrm{WSPT}}(I)$ is the completion time of $j$ and $C_j^{\mathrm{WSPT}}(0)(I) = S_j^{\mathrm{WSPT}}(I)$ is the starting time of $j$ in the WSPT schedule, and analogously for the optimal schedule. Furthermore, let $M_i^{\mathrm{WSPT}}(I)$ and $M_i^*(I)$ denote the load of the $i$-th machine and $M_{\mathrm{min}}^{\mathrm{WSPT}}(I)$ and $M_{\mathrm{min}}^{*}(I)$ denote the load of the least loaded machine, in the WSPT schedule and the optimal schedule for $I$, respectively. Moreover, let $\mathrm{WSPT}_\alpha(I)$ and $\mathrm{OPT}_\alpha(I)$ denote the weighted sum of $\alpha$-points of the schedule obtained by the WSPT rule and of the optimal schedule, respectively. Finally, denote by $\lambda_\alpha(I) \coloneqq \mathrm{WSPT}_\alpha(I)/\mathrm{OPT}_\alpha(I)$ the approximation ratio of the WSPT rule for the instance $I$. 
The WSPT rule does not specify which job to take first if multiple jobs have the same ratio $w_j/p_j$. Since we are interested in the performance of this rule in the worst case, we have to prove the performance ratio for all possible orders of these jobs. Hence, we may assume that this is done according to an arbitrary order of the jobs given as part of the input.

It is a well-known fact (see e.g.~\cite{Sch11}) that for the weighted sum of completion times objective the worst case of WSPT occurs if all jobs have the same Smith ratio $w_j/p_j$. This generalizes to the weighted sum of $\alpha$-points objective.

\begin{lem} \label{worst case unit ratio}
 For every $\alpha \in [0,1]$ and every instance $I$ of $P||\sum w_j C_j(\alpha)$ there is an instance $I'$ of $P||\sum p_j C_j(\alpha)$ with the same number of machines such that $\lambda_\alpha(I') \ge \lambda_\alpha(I)$.
\end{lem}

\begin{proof}
 The proof proceeds in the same way as Schwiegelshohn's proof \cite{Sch11}. Assume that $w_1/p_1 \ge w_2/p_j \ge \cdots \ge w_n/p_n$ and that the jobs are scheduled in this order in the WSPT schedule for $I$. Then define the instances $I'_k$, $k=1,\dotsc,n$ consisting of $k$ jobs with $w_j' = p_j' \coloneqq p_j$. Then for every $k \in \{1,\dotsc,n\}$ it holds that $C_j^{\mathrm{WSPT}}(\alpha)(I'_k) = C_j^{\mathrm{WSPT}}(\alpha)(I)$ for every job $j \in \{1,\dotsc,k\}$. Therefore, if we set $\rho_{n+1} \coloneqq 0$, we get
 \begin{align*}
  \mathrm{WSPT}_\alpha(I) &= \sum_{j=1}^n w_j C_j^{\mathrm{WSPT}}(\alpha)(I) = \sum_{j=1}^n \rho_j p_j C_j^{\mathrm{WSPT}}(\alpha)(I) \\ 
  &= \sum_{j=1}^n \left(\sum_{k=j}^n (\rho_k-\rho_{k+1})\right) p_j C_j^{\mathrm{WSPT}}(\alpha)(I) \\ 
  &= \sum_{k=1}^n (\rho_k-\rho_{k+1}) \sum_{j=1}^k p_j C_j^{\mathrm{WSPT}}(\alpha)(I) \\
  &= \sum_{k=1}^n (\rho_k-\rho_{k+1}) \sum_{j=1}^k p_j C_j^{\mathrm{WSPT}}(\alpha)(I'_k) = \sum_{k=1}^n (\rho_k-\rho_{k+1}) \mathrm{WSPT}_\alpha(I_k').
 \end{align*}
 On the other hand, for every $k \in \{1,\dotsc,n\}$, scheduling every job $j \in \{1,\dotsc,k\}$ as in the optimal schedule for $I$ is feasible for $I'_k$. Hence, we can bound the optimal objective value for $I_k$ by $\mathrm{OPT}_\alpha(I_k') \le \sum_{j=1}^k p_j C_j^*(\alpha)(I)$. Therefore,
 \begin{align*}
  \mathrm{OPT}_\alpha(I) &= \sum_{j=1}^n w_j C_j^*(\alpha)(I) = \sum_{j=1}^n \rho_j p_j C_j^*(\alpha)(I) \\
  &= \sum_{j=1}^n \left(\sum_{k=j}^n (\rho_k-\rho_{k+1})\right) p_j C_j^*(\alpha)(I) \\
  &= \sum_{k=1}^n (\rho_k - \rho_{k+1}) \sum_{j=1}^k p_j C_j^*(\alpha)(I) \ge \sum_{k=1}^n (\rho_k-\rho_{k+1}) \mathrm{OPT}_\alpha(I'_k)
 \end{align*}
 because $\rho_{k}-\rho_{k+1} \ge 0$ for all $k \in \{1,\dotsc,n\}$. Hence,
 \begin{multline*}
  \lambda_\alpha(I) = \frac{\mathrm{OPT}_\alpha(I)}{\mathrm{WSPT}_\alpha(I)} \le \frac{\sum_{k=1}^n (\rho_k-\rho_{k+1}) \mathrm{WSPT}_\alpha(I_k')}{\sum_{k=1}^n (\rho_k-\rho_{k+1}) \mathrm{OPT}_\alpha(I'_k)} \\ \le \max_{k \in \{1,\dotsc,n\}} \frac{\mathrm{WSPT}_\alpha(I'_k)}{\mathrm{OPT}_\alpha(I'_k)} = \max_{k \in \{1,\dotsc,m\}} \lambda_\alpha(I'_k). \qedhere\end{multline*}
\end{proof}

For unit Smith ratio instances the WSPT rule is nothing but list scheduling according to an arbitrary given order. Restricting to them has the benefit that the objective value of a schedule $\mathrm S$ can be computed easily from its machine loads, namely \begin{equation}\sum_{j=1}^n p_j C_j^{\mathrm S}(\tfrac 1 2) = \frac 1 2 \sum_{i=1}^m (M_i^{\mathrm S})^2.\label{eq:relation to machine loads}\end{equation} This classical observation can for example be found in the paper of Eastman et al.~\cite{EEI64}.

For the sum of the squares of the machine loads as objective function Avidor, Azar, and Sgall~\cite{AAS01} showed that WSPT has performance ratio $4/3$. So this also also holds for the weighted sum of $\tfrac 1 2$-points. By plugging it in into Theorem~\ref{Performance WSEPT based on alpha-points}, we get the following corollary. 

\begin{corollary} \label{orange performance guarantee}
 The WSEPT rule has performance ratio $1+\frac 1 6 \max\{2,1+\Delta\}$ for the problem~$P|\bm p_j \sim \mathrm{stoch}|\E[\sum w_j \bm C_j]$.
\end{corollary}

Now we generalize the bound of Eastman, Even, and Isaacs~\cite{EEI64}.

\begin{theorem}[Generalized Eastman-Even-Isaacs bound] \label{EEI bounds}
 For every $\alpha \in (0,1]$ the WSPT rule has performance ratio 
\[1+\frac{m-1}{2\alpha m} \le 1 + \frac{1}{2\alpha}\]
 for the problem $P||\sum w_j C_j(\alpha)$.
\end{theorem}
\begin{proof}
 By Lemma~\ref{worst case unit ratio} we can restrict to the case that $w_j = p_j$ for all jobs $j$. Let $I_1$ be the instance consisting of the same jobs but only one machine. On this instance the WSPT schedule is optimal.
 Moreover, the following two inequalities hold.
 \begin{align}
  \sum_{j=1}^n p_j S_j^{\mathrm{WSPT}}(I) &\le \frac 1 m \sum_{j=1}^n p_j S_j^{\mathrm{WSPT}}(I_1) = \frac 1 m \sum_{j=1}^n p_j S_j^*(I_1), \label{eq:bound start time list scheduling}\\
  \sum_{j-1}^n p_j C_j^*(\tfrac 1 2)(I) &\ge \frac 1 m \sum_{j=1}^n p_j C_j^*(\tfrac 1 2)(I_1) \label{eq:classical EEI bound}.
 \end{align}
 The first holds because a WSPT schedule is a list schedule and the starting time of any job~$j$ is at most the average machine load caused by all jobs preceding $j$ in the list. The second inequality, due to Eastman, Even, and Isaacs~\cite{EEI64}, follows from Equation~\eqref{eq:relation to machine loads} and the convexity of the square function. Putting these two inequalities together yields
 \begin{align*}
  &\sum_{j=1}^n p_j C_j^{\mathrm{WSPT}}(\alpha)(I) = \sum_{j=1}^n p_j S_j^{\mathrm{WSPT}}(I) + \alpha \sum_{j=1}^n p_j^2 \stackrel{\eqref{eq:bound start time list scheduling}}\le \frac 1 m \sum_{j=1}^n p_j S_j^{*}(I_1) + \alpha \sum_{j=1}^n p_j^2 \\
  =\ &\frac 1 m \sum_{j=1}^n p_j C_j^{*}(\tfrac 1 2)(I_1) + \Bigl(\alpha - \frac{1}{2m}\Bigr) \sum_{j=1}^n p_j^2 \stackrel{\eqref{eq:classical EEI bound}}\le \sum_{j=1}^n p_j C_j^*(\tfrac 1 2)(I) + \Bigl(\alpha - \frac{1}{2m}\Bigr) \sum_{j=1}^n p_j^2 \\
  =\ &\sum_{j=1}^n p_j C_j^*(\alpha)(I) + \frac{m-1}{2m} \sum_{j=1}^n p_j^2 \le \sum_{j=1}^n p_j C_j^*(\alpha)(I) + \frac{m-1}{2\alpha m} \sum_{j=1}^n p_j C_j^*(\alpha)(I) \\
  =\ &\left(1+\frac{m-1}{2\alpha m}\right) \sum_{j=1}^n p_j C_j^*(\alpha)(I). \qedhere
 \end{align*}
\end{proof}

\begin{remark}
 The generalized Eastman-Even-Isaacs bound does not lead to better performance ratios for the WSEPT rule for $P|\bm p_j \sim \mathrm{stoch}|\E[\sum w_j \bm C_j]$ than the bound of Möhring et al.~\cite{MSU99}, as plugging in $\beta = \frac{m-1}{2\alpha m}$ into Theorem~\ref{Performance WSEPT based on alpha-points} leads to a performance ratio of
 \[1 +  \frac{m-1}{2\alpha m} \cdot \max\{1,\alpha(1+\Delta)\} \ge 1 + \frac 1 2 (1+\Delta)\Bigl(1-\frac 1 m\Bigr).\]
\end{remark}

So far, by choosing $\alpha = 1$ and $\alpha = \tfrac 1 2$ we have derived the two performance ratios for the WSEPT rule labeled by [Cor.~\ref{red performance guarantee}] and [Cor.~\ref{orange performance guarantee}] in Figure~\ref{fig:performance}.
These are better than those following from Theorem~\ref{EEI bounds}. Besides, the proofs of Schwiegelshohn~\cite{Sch11} and of Avidor et al.~\cite{AAS01} of the underlying bounds for WSPT are quite similar. Both consist of a sequence of steps that reduce the set of instances to be examined. In every such reduction step it is shown that for any instance $I$ of the currently considered set there is an instance $I'$ in a smaller set for which the approximation ratio of WSPT is not better. This can be generalized to arbitrary $\alpha \in [\frac 1 2, 1]$. The resulting performance ratios for WSPT lead by means of Theorem~\ref{Performance WSEPT based on alpha-points} to a family of different performance ratios for the WSEPT rule. Note that the performance ratio of WSEPT following from the result of Avidor et al.\ for $\alpha= \tfrac 1 2$ has better behavior for large $\Delta$, while the performance ratio following from Kawaguchi and Kyan's result for $\alpha = 1$ is better for small $\Delta$. This behavior generalizes to $\alpha \in [\tfrac 1 2, 1]$: the smaller the underlying $\alpha$, the better the ratio for large $\Delta$ but the worse the ratio for small $\Delta$. Finally, we take for every $\Delta > 0$ the minimum of all the derived bounds.

\begin{theorem}[Generalized Kawaguchi-Kyan bound] \label{KK bounds}
 For every $\alpha \in [\frac 1 2, 1]$ the WSPT rule has performance ratio
 \[1+\frac{1}{2\alpha + \sqrt{8\alpha}}\]
 for $P||\sum w_j C_j(\alpha)$, and this bound is tight.
\end{theorem}

Combining this bound with Theorem~\ref{Performance WSEPT based on alpha-points} yields for every $\alpha \in [\frac 1 2, 1]$ the performance ratio $1+\tfrac 1 2 \max\{1/(\alpha+\sqrt{2\alpha}),(1+\Delta)/(1+\sqrt{2/\alpha})\}$ of WSEPT for $P|\bm p_j \sim \mathrm{stoch}|\E[\sum w_j \bm C_j]$. This is minimized at $\alpha \coloneqq 1/\min\{2,1+\Delta\}$, yielding the following performance ratio (see Figure~\ref{fig:performance}).

\begin{corollary} \label{green performance guarantee}
 For $P|\bm p_j \sim \mathrm{stoch}|\E[\sum w_j \bm C_j]$ the WSEPT rule has performance ratio
 \[1 + \frac 1 2 \cdot \frac{1}{1+\min\{2,\sqrt{2(1+\Delta)}\}} \cdot (1+\Delta).\]
\end{corollary}

\paragraph{Proof of Theorem~\ref{KK bounds}}
The proof of Theorem~\ref{KK bounds} is analogous to the proof of Schwiegelshohn~\cite{Sch11}, consisting of a sequence of lemmas that reduce the set of instances to consider until a worst-case instance is determined. From now on, let $\alpha \in [\tfrac 1 2, 1]$. Assuming that $p_1 \ge \cdots \ge p_n$, let \[\ell \coloneqq \max \left\{j \in \{1,\dotsc,m\} \,\middle|\, p_j \ge \frac{1}{m-j+1}\sum_{j'=j}^n p_{j'}\right\}.\] Then we call the $\ell$ jobs with largest processing times \emph{long} jobs and denote the set of long jobs by $\cL$. 

\begin{lem} \label{Schwiegelshohn2}
 For every instance $I$ of $P||\sum p_j C_j(\alpha)$ and every $\varepsilon > 0$ there is an instance $I' = I'(\varepsilon)$ of $P||\sum p_j C_j(\alpha)$ with the same number of machines such that $\lambda_\alpha(I') \ge \lambda_\alpha(I)$ and 
 \begin{enumerate}[(i)]
  \item $M_{\mathrm{min}}^{\mathrm{WSPT}}(I') = 1$,
  \item every job~$j$ with $S_j^{\mathrm{WSPT}}(I') < M_{\mathrm{min}}^{\mathrm{WSPT}}(I')$ fulfills $C_j^{\mathrm{WSPT}}(I') \le M_{\mathrm{min}}^{\mathrm{WSPT}}(I')$ and $p_j' < \varepsilon$,
  \item in the optimal schedule for~$I'$ every machine either is used only by a single long job or has load $M_{\mathrm{min}}^*(I')$.
 \end{enumerate}
\end{lem}
Like in Schwiegelshohn's paper, the lemma is proven by scaling the instance and splitting all jobs with $S_j^{\mathrm{WSPT}} < M_{\mathrm{min}}^{\mathrm{WSPT}}$ until they satisfy the conditions.
\begin{proof}
 As in the proof by Schwiegelshohn~\cite{Sch11}, the reduction relies on the observation that a job~$j$ with $S_j^{\mathrm{WSPT}} < M_{\mathrm{min}}^{\mathrm{WSPT}}$ can be replaced by two jobs $(j,1)$ and $(j,2)$ with $p_{(j,1)}'+p_{(j,2)}'=p_j$ and $S_j^{\mathrm{WSPT}} + p_{(j,1)}' \le M_{\mathrm{min}}^{\mathrm{WSPT}}$ in such a way that the WSPT rule schedules the new jobs one after the other on the same machine and during the same time slot as the old job and that the approximation ratio does not decrease. (Thereto the position of the second job in the input list must be chosen appropriately.) The reason that the approximation ratio can only increase is that the transformation reduces the objective value of the WSPT schedule by exactly \[\alpha p_j^2 - (\alpha p_{(j,1)}^2 + p_{(j,2)}'(p_{(j,1)}'+\alpha p_{(j,2)})) = (2\alpha -1) p_{(j,1)}'p_{(j,2)}' =: \delta \ge 0,\] and the objective value of the optimal schedule is reduced by at least this amount. (Replacing $j$ in the optimal schedule by the new jobs gives a feasible schedule). So
 \[\frac{\mathrm{WSPT}_\alpha - \delta}{\mathrm{OPT}_\alpha-\delta} = \lambda_\alpha + \frac{\delta}{\mathrm{OPT}_\alpha-\delta}(\lambda_\alpha - 1) \ge \lambda_\alpha\] because $\delta < \mathrm{OPT}_\alpha$ and $\lambda_\alpha \ge 1$.
 
 Such job splitting is applied to jobs with $S_j^{\mathrm{WSPT}} < M_{\mathrm{min}}^{\mathrm{WSPT}}$ until the conditions of the lemma are satisfied: First, every job jutting out over $M_{\mathrm{min}}^{\mathrm{WSPT}}$ in the WSPT schedule is split so that the first part ends at $M_{\mathrm{min}}^{\mathrm{WSPT}}$, then the jobs are split in such a way that they can be evenly distributed onto the machines without a long job in the optimal schedule, and finally, they are split until they are smaller than $M_{\mathrm{min}}^{\mathrm{WSPT}}(I) \cdot \varepsilon$. After the splitting, the whole instance is scaled by $M_{\mathrm{min}}^{\mathrm{WSPT}}(I)$ in order to fulfill the first condition.
\end{proof}
Note that the restriction to $\alpha \ge \tfrac 1 2$ is needed for this lemma because for smaller $\alpha$ splitting jobs increases the objective value and can thence reduce the performance ratio.

From now on, we focus on instances $I$ that fulfill the requirements of Lemma~\ref{Schwiegelshohn2} for some $0 < \varepsilon < M_{\mathrm{min}}^{\mathrm{WSPT}}(I)$. For a subset $\mathcal J \subseteq \mathcal N$ of jobs we write\linebreak $p(\mathcal J) \coloneqq \sum_{j \in \mathcal J} p_j$. We call the jobs in the set \[\cS \coloneqq \{j \in \{1,\dotsc,n\} \mid S_j^{\mathrm{WSPT}}(I) < M_{\mathrm{min}}^{\mathrm{WSPT}}\}\] \emph{short} jobs. This set is disjoint from $\cL$ because all jobs~$j \in \cS$ have processing time $p_j < \varepsilon$, while $p_j \ge p_\ell \ge \frac{1}{m-\ell+1}\sum_{j'=\ell}^n p_{j'} \ge M_{\mathrm{min}}^{\mathrm{WSPT}} > \varepsilon$ for all jobs~$j \in \cL$. Finally, we call the jobs in \[\cM \coloneqq \cN(I) \setminus (\cS \cup \cL)\] \emph{medium} jobs. In the optimal schedule for an instance $I$ of the type of Lemma~\ref{Schwiegelshohn2}, every machine~$i$ that does not process only a single long job has load\linebreak $M_i^*(I) = M_{\mathrm{min}}^{*}(I) = p(\cM \cup \cS)/(m-|\cL|)$. While some of them may process a medium job together with some short jobs, the rest only process short jobs (see Figure~\ref{Figure1}).

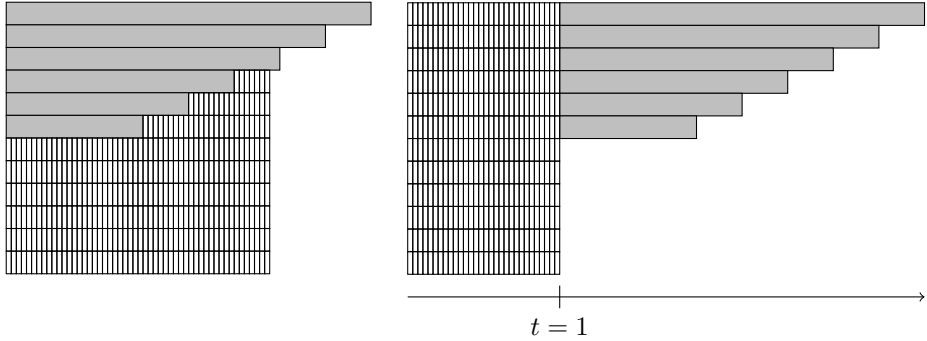
\begin{figure}
 \centering
 \begin{tikzpicture}[xscale=2,yscale=0.3]
  \filldraw[fill=lightgray] (0,11) rectangle +(2.4,1);
  \filldraw[fill=lightgray] (0,10) rectangle +(2.1,1);
  \filldraw[fill=lightgray] (0,9) rectangle +(1.8,1);
  \filldraw[fill=lightgray] (0,8) rectangle +(1.5,1);
  \filldraw[fill=lightgray] (0,7) rectangle +(1.2,1);
  \filldraw[fill=lightgray] (0,6) rectangle +(0.9,1);
  \foreach \j in {45,...,51} {
   \draw (\j/30,8) rectangle +(1/30,1);
  }
  \foreach \j in {36,...,51} {
   \draw (\j/30,7) rectangle +(1/30,1);
  }
  \foreach \j in {27,...,51} {
   \draw (\j/30,6) rectangle +(1/30,1);
  }
  \foreach \i in {0,...,5} {
   \foreach \j in {0,...,51} {
    \draw (\j/30,\i) rectangle +(1/30,1);
   }
  }
 \end{tikzpicture}
 \quad
 \begin{tikzpicture}[xscale=2,yscale=0.3,baseline]
  \filldraw[fill=lightgray] (1,11) rectangle +(2.4,1);
  \filldraw[fill=lightgray] (1,10) rectangle +(2.1,1);
  \filldraw[fill=lightgray] (1,9) rectangle +(1.8,1);
  \filldraw[fill=lightgray] (1,8) rectangle +(1.5,1);
  \filldraw[fill=lightgray] (1,7) rectangle +(1.2,1);
  \filldraw[fill=lightgray] (1,6) rectangle +(0.9,1);
  \foreach \i in {0,...,11} {
   \foreach \j in {0,...,29} {
    \draw (\j/30,\i) rectangle +(1/30,1);
   }
  }
  \draw[->] (0,-1) -- (3.4,-1);
  \draw (1,-0.5) -- (1,-1.5) node[below] {$t=1$};
 \end{tikzpicture}
 \caption{Optimal schedule and WSPT schedule for instance satisfying the conditions of Lemma~\ref{Schwiegelshohn2}} \label{Figure1}
\end{figure}

\begin{lem} \label{non-short jobs same length}
 For every instance $I$ of $P||\sum p_j C_j(\alpha)$ satisfying the conditions of Lemma~\ref{Schwiegelshohn2} there is an instance $I'$ with $\lambda_\alpha(I') \ge \lambda_\alpha(I)$ that still satisfies the conditions of Lemma~\ref{Schwiegelshohn2} and has the additional property that all non-short jobs require equal processing time.
\end{lem}
\begin{proof}
 The proof is an adapted version of the proof of Corollary 5 in the paper of Schwiegelshohn~\cite{Sch11}. We assume that processing times are rational. Let $A \coloneqq \sum_{j \in \mathcal L \cup \mathcal M} p_j^2$ and $t \coloneqq p(\cL \cup \cM)^2/A \in \Q$. Let $q \in \N$ be the denominator of $t$ (as reduced fraction). The instance $I'$ is defined as follows: The number $m'$ of machines is set to $qm$. Furthermore, there are $qt \in \N$ non-short jobs of size $p' \coloneqq A/p(\cL \cup \cM)$. Finally, for every short job of $I$ there are $q$ short jobs in $I'$. We assume that these can be distributed evenly onto all $qm$ machines and additionally, together with the medium jobs, they can be distributed in a balanced manner onto the machines that do not process a long job. (If this is not the case, we split the short jobs appropriately beforehand.) Then the conditions of Lemma~\ref{Schwiegelshohn2} remain valid and all non-short jobs have the same size. Let $\cS'$ be the set of short jobs in $I'$ and let $\mathcal{NS}'$ be the set of non-short jobs in $I'$. First notice that 
 \begin{equation*}
  \begin{array}{rcccl}
   \displaystyle \sum_{j \in \mathcal{NS}'} p_j' &=& qt \cdot p' &=& \displaystyle q\sum_{j \in \cL \cup \cM} p_j \quad \text{and} \\
   \displaystyle \sum_{j \in \mathcal{NS}'} (p_j')^2 &=& qt \cdot (p')^2 &=& \displaystyle q\sum_{j \in \cL \cup \cM} p_j^2.
  \end{array}
 \end{equation*}
 Therefore, the objective value of the WSPT schedule is only scaled by $q$ in the course of this transformation:
 \begin{align*}
  \mathrm{WSPT}_\alpha(I') &= \sum_{j \in \cS'} p_j' C_j^{\mathrm{WSPT}}(\alpha)(I') + \sum_{j \in \mathcal{NS}'} p_j'C_j^{\mathrm{WSPT}}(\alpha)(I') \\
  &= \sum_{j \in \cS'} p_j' C_j^{\mathrm{WSPT}}(\alpha)(I') + \sum_{j \in \mathcal{NS}'} p_j' + \alpha \sum_{j \in \mathcal{NS}'} (p_j')^2 \\
  &= q \left(\sum_{j \in \cS} p_j C_j^{\mathrm{WSPT}}(\alpha)(I) + \sum_{j \in \cL \cup \cM} p_j + \alpha \sum_{j \in \cL\cup\cM} p_j^2\right) \\ &= q \left(\sum_{j \in \cS} p_j C_j^{\mathrm{WSPT}}(\alpha)(I) + \sum_{j \in \cL \cup \cM} p_j C_j^{\mathrm{WSPT}}(\alpha)(I) \right) \\
  &= q \cdot \mathrm{WSPT}_\alpha(I).
 \end{align*}
 
 It remains to be shown that the objective value of an optimal schedule for $I'$ is at most $q$ times the optimal objective value for $I$. This follows from the following calculation, where the inequality marked with $(*)$ is proven below.
 \begin{align*}\mathrm{OPT}_\alpha(I') &= \sum_{j \in \cN(I')} p_j' C_j^{*}(\alpha)(I') = \sum_{j \in \cN(I')} p_j' C_j^{*}(\tfrac 1 2)(I') + \Bigl(\alpha-\frac 1 2\Bigr) \sum_{j \in \cN(I')} (p_j')^2 \\
  &\stackrel{\eqref{eq:relation to machine loads}}= \frac 1 2 \sum_{i=1}^{m'} M_i^{*}(I')^2 + \Bigl(\alpha - \frac 1 2\Bigr) \left(\sum_{j \in \cS'} (p_j')^2 + \sum_{j \in \mathcal{NS}'} (p_j')^2\right) \\
  &\stackrel{(*)}\le \frac 1 2 q \sum_{i=1}^m M_i^*(I)^2 + \Bigl(\alpha - \frac 1 2\Bigr) \left(q \sum_{j \in \cS} p_j^2 + q\sum_{j \in \cL \cup \cM} p_j^2\right) \\
  &\stackrel{\eqref{eq:relation to machine loads}}= q\left(\sum_{j \in \cN(I)} p_j C_j^*(\tfrac 1 2)(I) + \Bigl(\alpha - \frac 1 2\Bigr) \sum_{j \in \cN(I)} p_j^2\right) \\
  &= q \sum_{j \in \cN(I)} p_j C_j^*(\alpha)(I) = q \cdot \mathrm{OPT}_\alpha(I).
 \end{align*}
 In order to prove the inequality $(*)$, we have to show that
 \[\sum_{i=1}^{m'} M_i^*(I')^2 \le q \sum_{i=1}^m M_i^*(I)^2.\]
 
  If $p' \le 1+qtp'/m' = 1+tp'/m$, i.e., if the new jobs are medium jobs, then all machines have load $1+tp'/m$ in the optimal schedule for $I'$. Therefore,
  \begin{align*}
   \sum_{i=1}^{m'} M_i^{*}(I')^2 = m'(1+tp'/m)^2 &= qm \left(\frac{p(\cS) + p(\cL \cup \cM)}{m}\right)^2 \\ &= qm \left(\frac{\sum_{i=1}^m M_i^*(I)}{m}\right)^2 \le q \sum_{i=1}^m M_i^*(I)^2,
  \end{align*}
  where the last inequality follows from the convexity of the square function.

 In the other case, when the non-short jobs are long, there are $qt$ machines with load~$p'$ and $q(m-t)$ machines with load~$m/(m-t)$ in the optimal schedule for $I'$. We define $k \coloneqq p(\cM)/M_{\mathrm{min}}^*(I)$ and
 \[\tilde A \coloneqq \sum_{j \in \cL} p_j^2 + k \cdot M_{\mathrm{min}}^*(I)^2, \qquad \tilde t \coloneqq \frac{p(\cL \cup \cM)^2}{\tilde A}, \qquad \tilde p \coloneqq \frac{\tilde A}{p(\cL \cup \cM)}.\]
 The convexity of the square function implies that
 \begin{equation}
  \tilde t = (|\cL|+k) \cdot \frac{\left((|\cL|+k) \cdot \Bigl(\sum_{j \in \cL} p_j + k \cdot M_{\mathrm{min}}^*(I)\Bigr)/(|\cL|+k)\right)^2}{\sum_{j \in \cL} p_j^2 + k \cdot M_{\mathrm{min}}^*(I)^2} \le |\cL| + k \label{eq:bound tilde t}
 \end{equation}
 and
 \[k \cdot M_{\mathrm{min}}^*(I) = k \cdot \left(\frac{p(\cM)}{k}\right)^2 = \frac{p(\cM)^2}{k} \ge \frac{p(\cM)^2}{|\cM|} = |\cM| \left(\frac{p(\cM)}{|\cM|}\right)^2 \ge \sum_{j \in \cM} p_j^2,\]
 whence $\tilde A \ge A$, $\tilde t \le t$, and $\tilde p \ge p'$.
 Therefore,
 \[\frac{m}{m-\tilde t} \le \frac{m}{m-t} < p' \le \tilde p,\]
 where the inequality in the middle holds because by assumption $p' > 1+tp'/m$. Since additionally, $t \cdot p' + (m-t) \cdot \frac{m}{m-t} = p(\cN(I)) = \tilde t \cdot \tilde p + (m-\tilde t) \cdot \frac{m}{m-\tilde t}$, the convexity of the square function implies that
 \begin{equation}
  t \cdot (p')^2 + (m-t) \cdot \left(\frac{m}{m-t}\right)^2 \le \tilde t \cdot \tilde p^2 + (m-\tilde t) \cdot \left(\frac{m}{m-\tilde t}\right)^2. \label{eq:inequality Schwiegelshohn}
 \end{equation}
 For an illustration of this formula, see Schwiegelshohn~\cite{Sch11}. The following calculation concludes the proof.
 \begin{align*}
  \sum_{i=1}^{m'} M_i^*(I')^2 &= qt(p')^2 + q(m-t)\left(\frac{m}{m-t}\right)^2 = q\left(t(p')^2 + (m-t)\left(\frac{m}{m-t}\right)^2\right) \\ 
  &\stackrel{\eqref{eq:inequality Schwiegelshohn}}\le q\left(\tilde t \tilde p^2 + (m-\tilde t) \left(\frac{m}{m-\tilde t}\right)^2\right) \\
  &= q\left(\tilde t \tilde p^2 + \frac{m^2}{m-\tilde t}\right) \stackrel{\eqref{eq:bound tilde t}}\le q\left(\tilde t \tilde p^2 + \frac{m^2}{m-|\cL|-k} \right) \\
  &= q\left(\tilde A + (m-|\cL|-k) \left(\frac{m}{m-|\cL|-k}\right)^2\right) \\
  &= q\left(\sum_{j \in \cL} p_j^2 + k \cdot M_{\mathrm{min}}^*(I)^2 + (m-|\cL|-k) \cdot M_{\mathrm{min}}^*(I)^2\right) \\
  &= q\cdot \sum_{i=1}^m M_i^*(I)^2. \qedhere
 \end{align*}
\end{proof}

Since by Lemma~\ref{Schwiegelshohn2} reducing $\varepsilon$ can only increase the approximation ratio, the worst-case approximation ratio is approached in the limit $\varepsilon \to 0$, which we will subsequently further investigate. In the limit the sum of the squared processing times of the short jobs is negligible, wherefore the limits for $\varepsilon \to 0$ of the objective values of the WSPT schedule and the optimal schedule for an instance $I(\varepsilon)$ of the type of Lemma~\ref{non-short jobs same length} only depend on two variables: the ratio $s$ between the numbers of non-short jobs and machines and the duration $x$ of the non-short jobs. The limit of the objective value of the WSPT schedule is given by
\[\lim_{\varepsilon \to 0} \mathrm{WSPT}_\alpha(I(\varepsilon)) = \frac m 2 + smx(1+\alpha x).\]
For the optimal schedule the formula depends on whether the non-short jobs are medium or long. In the first case it is given by
\[\lim_{\varepsilon \to 0} \mathrm{OPT}_\alpha(I(\varepsilon)) = \frac m 2 (sx+1)^2 + \Bigl(\alpha - \frac 1 2\Bigr) smx^2.\]
and in the second case by
\[\lim_{\varepsilon \to 0} \mathrm{OPT}_\alpha(I(\varepsilon)) = \alpha s m x^2 + \frac{m}{2(1-s)}.\]
So we have to determine the maximum of the function
\[\lambda_\mathrm M(s,x) \coloneqq \frac{\frac m 2 + smx(1+\alpha x)}{\frac m 2 (sx+1)^2 + (\alpha - \frac 1 2) sm x^2} = \frac{2 s x (\alpha x + 1) + 1}{s^2 x^2 + sx((2\alpha - 1)x+2)+1}\] on $\{(s,x) \mid 0 \le s < 1, \ 0 \le x \le 1/(1-s)\}$ and the maximum of 
\[\lambda_\mathrm L(s,x) \coloneqq \frac{\frac m 2 + smx(1+\alpha x)}{\alpha s m x^2 + \frac{m}{2(1-s)}} = \frac{(1-s)(2sx(\alpha x+1)+1)}{2\alpha s(1-s)x^2 + 1}\] on the region $\{(s,x) \mid 0 \le s < 1,\ 1/(1-s) \le x\}$.

The partial derivative $\frac{\partial}{\partial x}\lambda_\mathrm M$ is positive on the feasible region, so for every fixed $s$ the maximum of $\lambda_\mathrm M(s,\cdot)$ is attained at $x=\frac{1}{1-s}$, corresponding to the case that the non-short jobs are long. This case is also captured by the function~$\lambda_\mathrm L$.

For $x \to \infty$ the function $\lambda_\mathrm L$ converges to one. Hence, for every $s$ the maximum of $\lambda_\mathrm L(s,\cdot)$ must be attained at a finite point $x$. The partial derivative $\frac{\partial}{\partial x}\lambda_\mathrm L$ has only one positive root, namely \[x_s \coloneqq (\alpha s + \sqrt{(2(1-s)+\alpha s)\alpha s})/(2\alpha s (1-s)) > 1/(1-s).\] By plugging this in, we obtain
\[\lambda_\mathrm L(s,x_s) = 1 + \tfrac 1 2 (\sqrt{(2(1-s)+\alpha s)\alpha s}/\alpha-s).\]
The only root of the derivative of the function $s \mapsto \lambda_\mathrm L(s,x_s)$ that is less than 1 is \[s \coloneqq 1/(2+\sqrt{2\alpha}).\] Plugging this in yields the worst-case performance ratio
\[1+\frac{1}{2\alpha+\sqrt{8\alpha}}.\]

Like the proofs of Kawaguchi and Kyan~\cite{KK86} and of Avidor et al.~\cite{AAS01}, this proof shows how the worst-case instances look like: They consist of short jobs of total length $m$ and $1/(2+\sqrt{2\alpha})m$ long jobs of length $1+\sqrt{2/\alpha}$. For $\alpha \in \{1/2,1\}$ we recover the worst case instances of Avidor et al.\ and of Kawaguchi and Kyan.

\section{Performance ratio of the WSPT rule for a fixed number of machines} \label{Performance guarantee WSPT}

In this section we analyze the WSPT rule for the problem $P||\sum w_j C_j$ with a fixed number $m$ of machines. The problem instances of Kawaguchi and Kyan~\cite{KK86} whose approximation ratios converge to $(1+\sqrt 2)/2$ consist of a set of infinitesimally short jobs with total processing time $m$, and a set of $k$ jobs of length $1+\sqrt 2$, where $k/m \to 1-\sqrt 2/2$. Since $1-\sqrt 2/2$ is irrational, the worst case ratio can only be approached if the number of machines goes to infinity. Rounding these instances for a fixed $m$ by choosing $k$ as the nearest integer to $\bigl(1-\frac{\sqrt 2}{2}\bigr)m$ (in the following denoted by~$\bigl\lfloor\bigl(1-\frac{\sqrt 2}{2}\bigr)m\bigr\rceil$) yields at least a lower bound on the worst-case approximation ratio for $P||\sum w_j C_j$. It is, however, possible that for a particular $m$ there is an instance for which the WSPT rule has a larger approximation ratio than for the rounded instance of Kawaguchi and Kyan. 
As we will see, the worst-case instances for any fixed $m$ actually look almost as the rounded Kawaguchi-Kyan instances with the only difference that the length of the long jobs depends as a function on $m$.

\begin{theorem} \label{theorem:Performance guarantee WSPT}
For $P||\sum w_j C_j$ the WSPT rule has performance ratio 
\[1+\frac 1 2 \frac{\sqrt{(2m-k_m)k_m}-k_m}{m},\qquad\text{where}\quad k_m \coloneqq \left\lfloor \Bigl(1-\frac{\sqrt 2}{2}\Bigr)m \right\rceil.\]
Moreover, this bound is tight for every fixed $m \in \N$.
\end{theorem}

In the remainder we prove this theorem. Lemmas~\ref{worst case unit ratio} and \ref{Schwiegelshohn2} hold in particular for the weighted sum of completion times. Since the described transformations do not change the number of machines, also for a fixed number $m$ of machines the worst case occurs in an instance of the form described in Lemma~\ref{Schwiegelshohn2}. However, we cannot apply Lemma~\ref{non-short jobs same length} when $m$ is fixed because the transformation in this lemma possibly changes the number of machines. As this is not allowed in our setting, we have to find different reductions. We first reduce to instances with at most one medium job and then reduce further to instances where all long jobs have equal length. Similar reductions are also carried out by Kalaitzis, Svensson, and Tarnawski~\cite{KST17}.

\begin{lem} \label{one medium job}
 For every instance $I$ of $P||\sum p_j C_j$ satisfying the conditions of Lemma~\ref{Schwiegelshohn2} there is an instance $I'$ with the same number of machines and\linebreak $\lambda(I') \ge \lambda(I)$ that still satisfies the conditions of Lemma~\ref{Schwiegelshohn2} and has the additional property that there is at most one medium job.
\end{lem}
\begin{proof}
 We replace the set $\cM$ of medium jobs in $I$ by $\lfloor p(\cM)/M_{\mathrm{min}}^{*}(I) \rfloor$ jobs of size $M_{\mathrm{min}}^{*}(I)$ and one job of size $p^\prime \coloneqq p(\cM) - M_{\mathrm{min}}^{*}(I) \lfloor p(\cM)/M_{\mathrm{min}}^{*}(I) \rfloor$ (see Figure~\ref{Figure2}).
In the thus defined instance $I^\prime$ the job of size $p^\prime$ is the only potentially medium job. Moreover, the properties of Lemma~\ref{Schwiegelshohn2} remain true.
 Let $\tilde \cM$ be the set of new jobs in the instance $I^\prime$. Then we have
 \begin{align*}
  \sum_{j \in \cN(I')} p_j' - \sum_{j \in \cN(I)} p_j = \sum_{j \in \tilde \cM} p_j' - \sum_{j \in \cM} p_j &= 0 \quad \text{and} \\
  \sum_{j \in \cN(I')} (p_j')^2 - \sum_{j \in \cN(I)} p_j^2 = \sum_{j \in \tilde \cM} (p_j')^2 - \sum_{j \in \cM} p_j^2 &=: \delta > 0
 \end{align*}
 by the convexity of the square function. 
 Since in the WSPT schedules for both problem instances all jobs in $\cM$ resp.\ $\tilde \cM$ have the same starting time $M_{\mathrm{min}}^{\mathrm{WSPT}} \coloneqq M_{\mathrm{min}}^{\mathrm{WSPT}}(I) = M_{\mathrm{min}}^{\mathrm{WSPT}}(I^\prime)$, and all other jobs remain unchanged, we have
 \begin{align*}
  \sum_{j \in \cN(I')} p_j' S_j^{\mathrm{WSPT}}(I') &= \sum_{j \in \cN(I') \setminus \tilde \cM} p_j' S_j^{\mathrm{WSPT}}(I') + \sum_{j \in \tilde \cM} p_j' S_j^{\mathrm{WSPT}}(I') \\
  &= \sum_{j \in \cN(I') \setminus \tilde \cM} p_j' S_j^{\mathrm{WSPT}}(I') + M_{\mathrm{min}}^{\mathrm{WSPT}} \sum_{j \in \tilde \cM} p_j' \\
  &= \sum_{j \in \cN(I) \setminus \cM} p_j S_j^{\mathrm{WSPT}}(I) + M_{\mathrm{min}}^{\mathrm{WSPT}} \sum_{j \in \cM} p_j \\
  &= \sum_{j \in \cN(I) \setminus \cM} p_j S_j^{\mathrm{WSPT}}(I) + \sum_{j \in \cM} p_j S_j^{\mathrm{WSPT}}(I) \\
  &= \sum_{j \in \cN(I)} p_j S_j^{\mathrm{WSPT}}(I).
 \end{align*}
 Therefore, we get
 \begin{multline*}
  \mathrm{WSPT}(I') = \sum_{j \in \cN(I')} p_j' C_j^{\mathrm{WSPT}}(I') = \sum_{j \in \cN(I')} p_j' S_j^{\mathrm{WSPT}}(I') + \sum_{j \in \cN(I')} (p_j')^2 \\ = \sum_{j \in \cN(I)} p_j S_j^{\mathrm{WSPT}}(I) + \sum_{j \in \cN(I)} p_j^2 + \delta = \sum_{j \in \cN(I)} p_j C_j^{\mathrm{WSPT}}(I) + \delta = \mathrm{WSPT}(I) + \delta.
 \end{multline*}
 After the transformation the schedule that puts every long job on a machine of its own and balances the loads of the remaining machines is still optimal. In this schedule all machines have the same load as in the optimal schedule for $I$. Therefore, we have 
 \begin{multline*} 
  \mathrm{OPT}(I') = \sum_{j \in \cN(I')} p_j' C_j^*(I') \stackrel{\eqref{eq:relation to machine loads}}= \frac 1 2 \left(\sum_{i=1}^m (M_i^*(I'))^2 + \sum_{j \in \cN(I')} (p_j')^2 \right) \\
  = \frac 1 2 \left( \sum_{i=1}^m (M_i^{*}(I))^2 + \sum_{j \in \cN(I)} p_j^2 + \delta\right) \stackrel{\eqref{eq:relation to machine loads}}= \sum_{j \in \cN(I)} p_j C_j^*(I) + \frac 1 2 \delta = \mathrm{OPT}(I) + \delta/2.
 \end{multline*} 
 Together this yields
 \[\lambda(I') = \frac{\mathrm{WSPT}(I')}{\mathrm{OPT}(I')} = \frac{\mathrm{WSPT}(I)+\delta}{\mathrm{OPT}(I)+\delta/2} = \lambda(I) + \frac{\delta}{\mathrm{OPT}(I)+\delta/2} \Bigl(1-\frac{\lambda(I)}{2}\Bigr).\]
 As $\lambda(I) < 2$, the second summand is non-negative, implying that $\lambda(I') \ge \lambda(I)$.
\end{proof}
For the instance shown in Figure~\ref{Figure1} the optimal and the WSPT schedule of the reduced instance are shown in Figure~\ref{Figure2}.
\begin{figure}[t]
 \centering
 \begin{tikzpicture}[xscale=2,yscale=0.3]
  \filldraw[fill=lightgray] (0,11) rectangle +(2.4,1);
  \filldraw[fill=lightgray] (0,10) rectangle +(2.1,1);
  \filldraw[fill=lightgray] (0,9) rectangle +(1.8,1);
  \filldraw[fill=lightgray] (0,8) rectangle +(26/15,1);
  \filldraw[fill=lightgray] (0,7) rectangle +(26/15,1);
  \filldraw[fill=lightgray] (0,6) rectangle +(2/15,1);
  \foreach \j in {4,...,51} {
   \draw (\j/30,6) rectangle +(1/30,1);
  }
  \foreach \i in {0,...,5} {
   \foreach \j in {0,...,51} {
    \draw (\j/30,\i) rectangle +(1/30,1);
   }
  }
 \end{tikzpicture}
 \quad
 \begin{tikzpicture}[xscale=2,yscale=0.3,baseline]
  \filldraw[fill=lightgray] (1,11) rectangle +(2.4,1);
  \filldraw[fill=lightgray] (1,10) rectangle +(2.1,1);
  \filldraw[fill=lightgray] (1,9) rectangle +(1.8,1);
  \filldraw[fill=lightgray] (1,8) rectangle +(26/15,1);
  \filldraw[fill=lightgray] (1,7) rectangle +(26/15,1);
  \filldraw[fill=lightgray] (1,6) rectangle +(2/15,1);
  \foreach \i in {0,...,11} {
   \foreach \j in {0,...,29} {
    \draw (\j/30,\i) rectangle +(1/30,1);
   }
  }
  \draw[->] (0,-1) -- (3.4,-1);
  \draw (1,-0.5) -- (1,-1.5) node[below] {$t=1$};
 \end{tikzpicture}
 \caption{Optimal schedule and WSPT schedule for instance after the transformation of Lemma~\ref{one medium job}} \label{Figure2}
\end{figure}
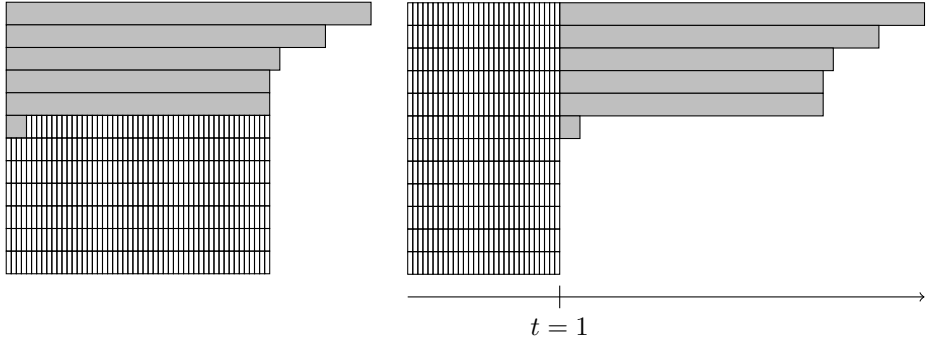

\begin{lem} \label{long jobs equal}
 For every instance $I$ of $P||\sum p_j C_j$ satisfying the conditions of Lemma~\ref{one medium job} there is an instance $I'$ with the same number of machines and\linebreak $\lambda(I') \ge \lambda(I)$ that still satisfies the conditions of Lemma~\ref{one medium job} where additionally all long jobs have equal processing time.
\end{lem}
\begin{proof}
 We replace the set $\cL$ of long jobs by $|\cL|$ jobs with processing time $p_j^\prime \coloneqq p(\cL)/|\cL|$. Then these jobs are still long, and no other jobs become long during this transformation, so the set $\cL'$ of long jobs in $I'$ consists exactly of the newly defined jobs. Besides, the sets of medium and short jobs are not affected by this transformation. The WSPT schedule for $I^\prime$ schedules the jobs in the same way as for $I$. Similarly, it is still optimal to schedule each long job on a machine of its own, and totally balance the loads of the remaining machines. Therefore, $I^\prime$ still satisfies the conditions of Lemma~\ref{one medium job}. The optimal and the WSPT schedule after the transformation are depicted in Figure~\ref{Figure3}.

 Note that
 \begin{align*}
  \sum_{j \in \cN(I')} p_j' - \sum_{j \in \cN(I)} p_j = \sum_{j \in \cL'} p_j' - \sum_{j \in \cL} p_j &= 0 \quad \text{and} \\
  \sum_{j \in \cN(I')} (p_j')^2 - \sum_{j \in \cN(I)} p_j^2 = \sum_{j \in \cL'} (p_j')^2 - \sum_{j \in \cL} p_j^2 &=: -\delta < 0,
 \end{align*}
 where the second inequality follows again from the convexity of the square function.
 
 Because in in the WSPT schedules for $I$ and $I'$ all long jobs start at time $M_{\mathrm{min}}^{\mathrm{WSPT}} \coloneqq M_{\mathrm{min}}^{\mathrm{WSPT}}(I) = M_{\mathrm{min}}^{\mathrm{WSPT}}(I^\prime)$ and all other jobs are unmodified, the same calculation as in the proof of Lemma~\ref{one medium job} shows that
 \[\sum_{j \in \cN(I')} p_j'S_j^{\mathrm{WSPT}}(I') = \sum_{j \in \cN(I)} p_j S_j^{\mathrm{WSPT}}(I),\]
 and thus,
 \begin{multline*}
  \mathrm{WSPT}(I') = \sum_{j \in \cN(I')} p_j'C_j^{\mathrm{WSPT}}(I') = \sum_{j \in \cN(I')} p_j'S_j^{\mathrm{WSPT}}(I') + \sum_{j \in \cN(I')} (p_j')^2 \\ = \sum_{j \in \cN(I)} p_j S_j^{\mathrm{WSPT}}(I) + \sum_{j \in \cN(I)} p_j^2 - \delta = \sum_{j \in \cN(I)} p_j C_j^{\mathrm{WSPT}}(I) - \delta = \mathrm{WSPT}(I) - \delta.
 \end{multline*}
 In the optimal schedule for $I$ and $I'$ all long jobs start at time $0$, so we also have that
 \[\sum_{j \in \cN(I')} p_j'S_j^{*}(I') = \sum_{j \in \cN(I)} p_j S_j^{*}(I),\]
 and hence,
 \begin{multline*}
  \mathrm{OPT}(I') = \sum_{j \in \cN(I')} p_j'C_j^{*}(I') = \sum_{j \in \cN(I')} p_j'S_j^{*}(I') + \sum_{j \in \cN(I')} (p_j')^2 \\ = \sum_{j \in \cN(I)} p_j S_j^{*}(I) + \sum_{j \in \cN(I)} p_j^2 - \delta = \sum_{j \in \cN(I)} p_j C_j^{*}(I) - \delta = \mathrm{OPT}(I) - \delta.
 \end{multline*}
 Together, this results in the inequality
 \[\lambda(I') = \frac{\mathrm{WSPT}(I')}{\mathrm{OPT}(I')} = \frac{\mathrm{WSPT}(I) - \delta}{\mathrm{OPT}(I) - \delta} = \lambda(I) + \frac{\delta}{\mathrm{OPT}(I)-\delta}(\lambda(I)-1) \ge \lambda(I)\]
 because $\delta < \mathrm{OPT}(I)$.
\end{proof}

The reduction used in the proof is illustrated in Figure~\ref{Figure3}.
\begin{figure}[t]
 \centering
 \begin{tikzpicture}[xscale=2,yscale=0.3]
  \filldraw[fill=lightgray] (0,11) rectangle +(293/150,1);
  \filldraw[fill=lightgray] (0,10) rectangle +(293/150,1);
  \filldraw[fill=lightgray] (0,9) rectangle +(293/150,1);
  \filldraw[fill=lightgray] (0,8) rectangle +(293/150,1);
  \filldraw[fill=lightgray] (0,7) rectangle +(293/150,1);
  \filldraw[fill=lightgray] (0,6) rectangle +(2/15,1);
  \foreach \j in {4,...,51} {
   \draw (\j/30,6) rectangle +(1/30,1);
  }
  \foreach \i in {0,...,5} {
   \foreach \j in {0,...,51} {
    \draw (\j/30,\i) rectangle +(1/30,1);
   }
  }
 \end{tikzpicture}
 \quad
 \begin{tikzpicture}[xscale=2,yscale=0.3,baseline]
  \filldraw[fill=lightgray] (1,11) rectangle +(293/150,1);
  \filldraw[fill=lightgray] (1,10) rectangle +(293/150,1);
  \filldraw[fill=lightgray] (1,9) rectangle +(293/150,1);
  \filldraw[fill=lightgray] (1,8) rectangle +(293/150,1);
  \filldraw[fill=lightgray] (1,7) rectangle +(293/150,1);
  \filldraw[fill=lightgray] (1,6) rectangle +(2/15,1);
  \foreach \i in {0,...,11} {
   \foreach \j in {0,...,29} {
    \draw (\j/30,\i) rectangle +(1/30,1);
   }
  }
  \draw[->] (0,-1) -- (3.4,-1);
  \draw (1,-0.5) -- (1,-1.5) node[below] {$t=1$};
 \end{tikzpicture}
 \caption{Optimal schedule and WSPT schedule for instance after the transformation of Lemma~\ref{long jobs equal}} \label{Figure3}
\end{figure}
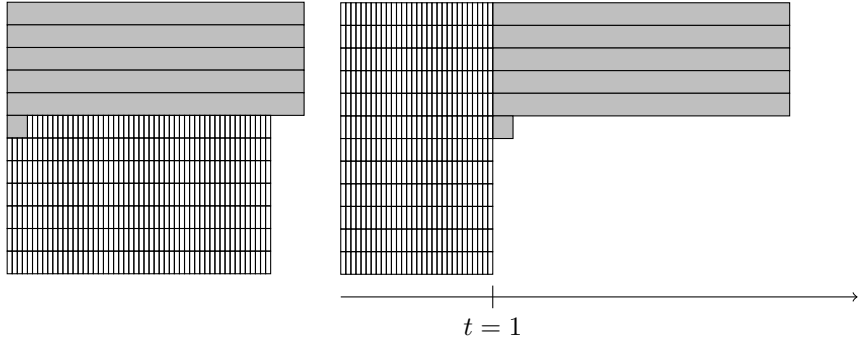

As in Section~\ref{Performance WSEPT based on alpha-points} we will analyze the limit for $\varepsilon \to 0$. 
The limits of the objective values of the WSPT schedule and the optimal schedule for an instance $I(\varepsilon)$ of the type of Lemma~\ref{long jobs equal} depend only on three variables: two real variables, viz.\ the length $x$ of the long jobs and the length $y$ of the medium job ($y=0$ if no medium job exists), and one integer variable, namely the number $k$ of long jobs. They are given by
\begin{align*}
 \lim_{\varepsilon \to 0} \mathrm{WSPT}(I(\varepsilon)) &= \frac m 2 + kx(1+x) + y(1+y),\\
 \lim_{\varepsilon \to 0} \mathrm{OPT}(I(\varepsilon)) &= k \cdot x^2 + \frac{(m+y)^2}{2(m-k)} + \frac{y^2}{2}.
\end{align*}
In Figure~\ref{Gantt charts}
\begin{figure}
 \begin{tikzpicture}[scale=2]
  \begin{scope}[yshift=2.1cm]
   \draw[<->] (2.1,0) -- (0,0) -- (0,2.1);
   \filldraw[fill=lightgray] (0,0) rectangle node {$x^2$} +(293/150,293/150);
   \node[anchor=east] at (0,293/300) {\footnotesize $k\ \times\ $};
  \end{scope}
  \begin{scope}
   \draw[<->] (1.9,0) -- (0,0) -- (0,1.9);
   \filldraw[fill=lightgray] (0,8/5) rectangle +(2/15,2/15);
   \draw[pattern color=gray, pattern=north east lines] (2/15,8/5) -- (26/15,0) -- (0,0) -- (0,8/5) -- cycle;
   \node[anchor=east] at (0,13/15) {\footnotesize $1\ \times\ $};
   \node at (29/45,15/45) {\footnotesize $\tfrac 1 2 \Bigl(\bigl(\frac{m+y}{m-k}\bigr)^2-y^2\Bigr)$};
   \node at (0.25,5/3) {$y^2$};
  \end{scope}
  \begin{scope}[yshift=-2.1cm]
   \draw[<->] (1.9,0) -- (0,0) -- (0,1.9);
   \draw[pattern color=gray, pattern=north east lines] (0,26/15) -- (26/15,0) -- (0,0) -- cycle;
   \node[anchor=east] at (0,13/15) {\footnotesize $(m-k-1)\ \times\ $};
   \node at (26/45,26/45) {$\frac 1 2 \left(\frac{m+y}{m-k}\right)^2$};
  \end{scope}
  \begin{scope}[yshift=1cm,xshift=2.5cm]
   \draw[<->] (3.1,0) -- (0,0) -- (0,3.1);
   \draw[pattern color=gray,pattern=north east lines] (0,443/150) -- (1,293/150) -- (0,293/150) -- cycle;
   \filldraw[fill=lightgray] (1,0) rectangle node {$x^2$} +(293/150,293/150);
   \node[anchor=east] at (0,443/300) {\footnotesize $k\ \times\ $};
   \node at (1/2,293/300) {$x$};
   \node at (1/3,343/150) {$\frac 1 2$};
   \draw (1,0.05) -- (1,-0.05) node[below] {$1$};
  \end{scope}
  \begin{scope}[yshift=-.65cm,xshift=2.5cm]
   \draw[<->] (1.3,0) -- (0,0) -- (0,1.3);
   \draw[pattern color=gray,pattern=north east lines] (0,17/15) -- (1,2/15) -- (0,2/15) -- cycle;
   \filldraw[fill=lightgray] (1,0) rectangle +(2/15,2/15);
   \node[anchor=east] at (0,17/30) {\footnotesize $1\ \times\ $};
   \node at (1/2,1/15) {$y$};
   \node at (1/3,7/15) {$\frac 1 2$};
   \node at (1.25,0.1) {$y^2$};
   \draw (1,0.05) -- (1,-0.05) node[below] {$1$};
  \end{scope}
  \begin{scope}[yshift=-2.1cm,xshift=2.5cm]
   \draw[<->] (1.1,0) -- (0,0) -- (0,1.1);
   \draw[pattern color=gray,pattern=north east lines] (0,1) -- (1,0) -- (0,0) -- cycle;
   \node[anchor=east] at (0,1/2) {\footnotesize $(m-k-1)\ \times\ $};
   \node at (1/3,1/3) {$\frac 1 2$};
   \draw (1,0.05) -- (1,-0.05) node[below] {$1$};
  \end{scope}
 \end{tikzpicture}
 \caption{Limit two-dimensional Gantt charts (see~\cite{EEI64,GW00}) for $\varepsilon \to 0$ of the optimal schedule and the WSPT schedule for instance $I(\varepsilon)$ satisfying the conditions of Lemma~\ref{long jobs equal}, illustrating the different parts of the objective values} \label{Gantt charts}
\end{figure}
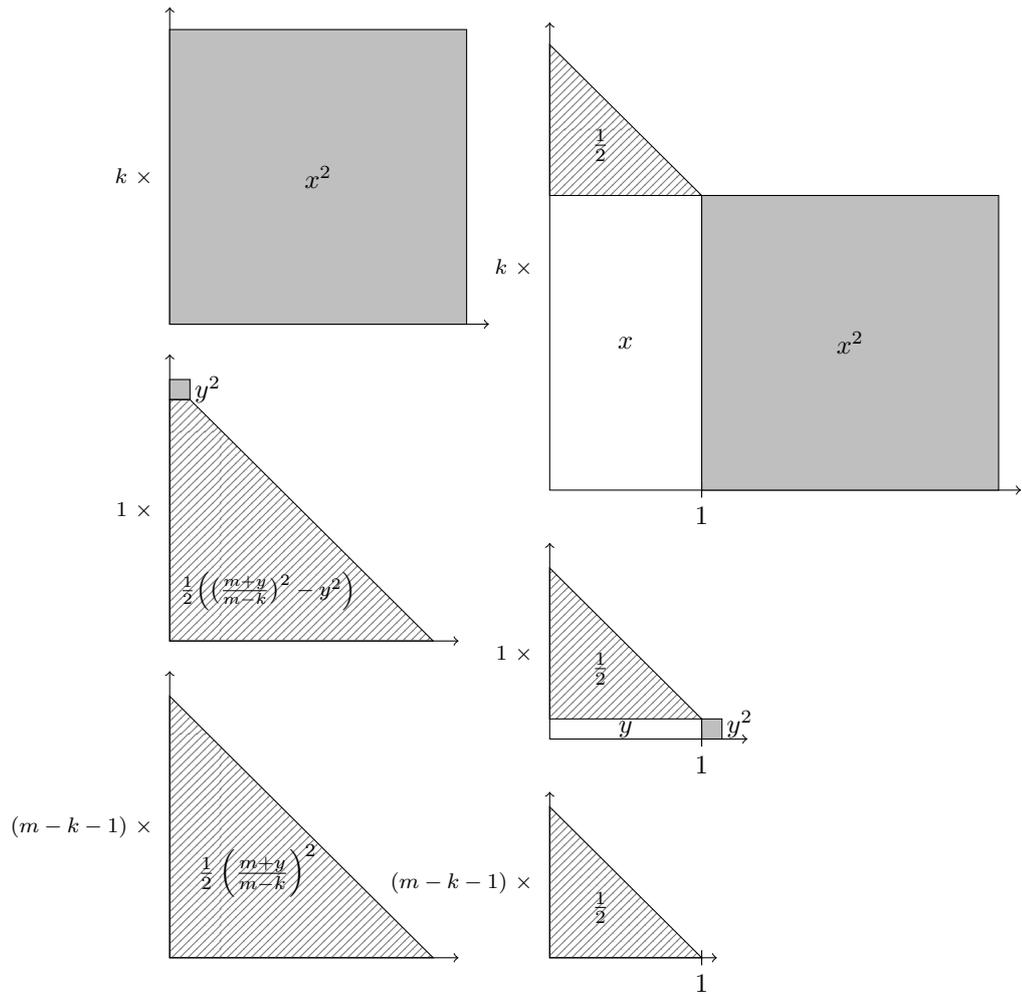
these formulas are illustrated via two-dimensional Gantt charts (see~\cite{EEI64,GW00})
for the three different types of single-machine schedules used by the WSPT schedule and the optimal schedule, respectively. In order to describe a valid scheduling instance of the prescribed type, the values $x$, $y$, and $k$ must lie in the domains
\[
 k \in \{0,\dotsc,m-1\},\qquad
 y \begin{cases} \in \Bigl[0,\frac{m}{m-k-1}\Bigr] &\text{if } k < m-1, \\ =0 &\text{if } k=m-1,\end{cases}\qquad
 x \in \left[\frac{y+m}{m-k},\infty\right).\]
\begin{lem} \label{lem:maximum}
The maximum of the ratio
\[
\lambda_m(x,y,k) \coloneqq \frac{\frac m 2 + kx(1+x) + y(1+y)}{k \cdot x^2 + \frac{(m+y)^2}{2(m-k)} + \frac{y^2}{2}} =  \frac{(m-k)(2kx^2+2kx+2y^2+2y+m)}{(m-k)(2kx^2+y^2)+(y+m)^2}
\] 
on the feasible domains is $1+\tfrac 1 2 (\sqrt{(2m-k_m)k_m}-k_m)/m$, and it is attained at
\[
 k_m \coloneqq \left\lfloor\Bigl(1-\frac12\sqrt 2\Bigr) m \right\rceil,\qquad
 y_m \coloneqq 0, \qquad
 x_m \coloneqq \frac{m}{\sqrt{(2m-k_m)k_m}-k_m}.
\]
\end{lem}
\begin{proof}
This set of possible combinations is closed. Since for $x \to \infty$ this ratio converges to $1$, the maximum is attained at some point~$(x^*,y^*,k^*)$.

Assume now that $y^* \neq 0$. Then $k^* < m-1$, and the function $y \mapsto \lambda(x^*,y,k^*)$ is thus quasi-convex (since the derivative has at most one root in the feasible region and is non-positive at $y=0$). Therefore, the maximum is attained at the boundary, i.e., $y^* = m/(m-k-1)$. But then the job is in fact a long job, and by Lemma~\ref{long jobs equal} it has the same length as the other long jobs, i.e., $x=y=m/(m-k-1)$. By the definitions of $k$, $x$, and $y$, however, this instance is properly described by the parameter values $k' = k+1$, $y' = 0$, and $x' = x$.
So we have shown that $y^* = 0$, and only the values $x$ and $k$ maximizing
\[\lambda_m'(x,k) \coloneqq \frac{(m-k)(2kx^2+2kx+m)}{(m-k)2kx^2+m^2}\]
remain to be determined.
For every fixed $k$, the maximum is attained at
\[x_{k,m} \coloneqq \frac{m}{\sqrt{(2m-k)k}-k},\]
as can be seen by calculating the roots of the derivative. Plugging this in, we obtain the univariate function
\[\lambda_m''(k) \coloneqq 1+\frac{\sqrt{(2m-k)k}-k}{2m},\]
whose maximum over $\{0,\dotsc,m-1\}$ is to be determined. Notice that~$\lambda_m''(0)=1$ is not maximal.
As a function on the interval $[1,m-1]$, the function $\lambda_m''$ is concave (since the second derivative is negative) and continuous. Furthermore,
\[u_m \coloneqq m-\frac{\sqrt{2m^2-1}}{2}\]
satisfies the equation $\lambda_m''(u_m-\frac 1 2) = \lambda_m''\bigl(u_m+\frac 1 2\bigr)$. Therefore, for every\linebreak $\xi \in \bigl[u_m-\frac 1 2,u_m+\frac 1 2\bigr)$ and every $\zeta \in [1,m-1] \setminus \bigl[u_m-\frac 1 2,u_m+\frac 1 2\bigr)$ it holds that $\lambda_m''(\xi) \ge \lambda_m''(\zeta)$. In particular, the unique integer $k$ in the interval\linebreak $\bigl[u_m-\frac 1 2,u_m+\frac 1 2\bigr)$ satisfies that $\lambda_m''(k) \ge \lambda_m''(k')$ for all $k' \in \{1,\dotsc,m-1\}$. We prove that
\[
k_m \coloneqq \left\lfloor \Bigl(1 - \frac{\sqrt 2}{2}\Bigr) m \right\rceil
\]
always lies in this interval. Firstly, $k_m \le \bigl(1 - \frac{\sqrt 2}{2}\bigr)m + \frac 1 2 < u_m + \frac 1 2$. Secondly, notice that $\sqrt{2m^2}$ is irrational such that $\left\lceil\sqrt{2m^2}\right\rceil\le\sqrt{2m^2-1}+1$ as the interval $\left(\sqrt{2m^2-1},\sqrt{2m^2}\right]$ does not contain an integer; therefore,
\[
k_m = m-\left\lfloor\frac{\left\lceil\sqrt{2m^2}\right\rceil}{2}\right\rfloor 
\ge m - \left\lfloor \frac{\sqrt{2m^2-1}+1}{2} \right\rfloor 
\ge m - \frac{\sqrt{2m^2-1}+1}{2} = u_m - \frac 1 2 .
\]
Figure~\ref{graph of f}
\begin{figure}
 \includegraphics[width=.45\textwidth]{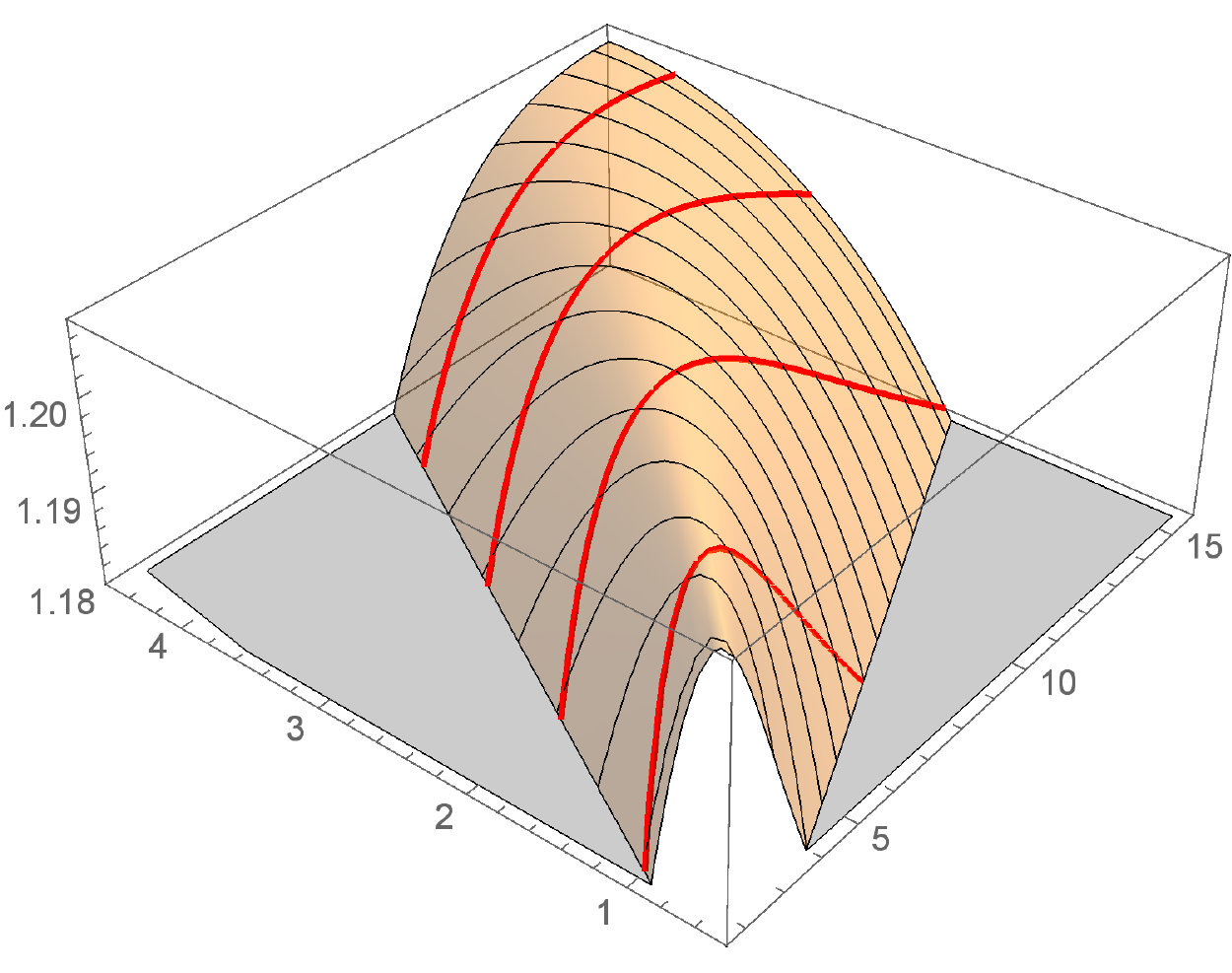} \includegraphics[width=.45\textwidth]{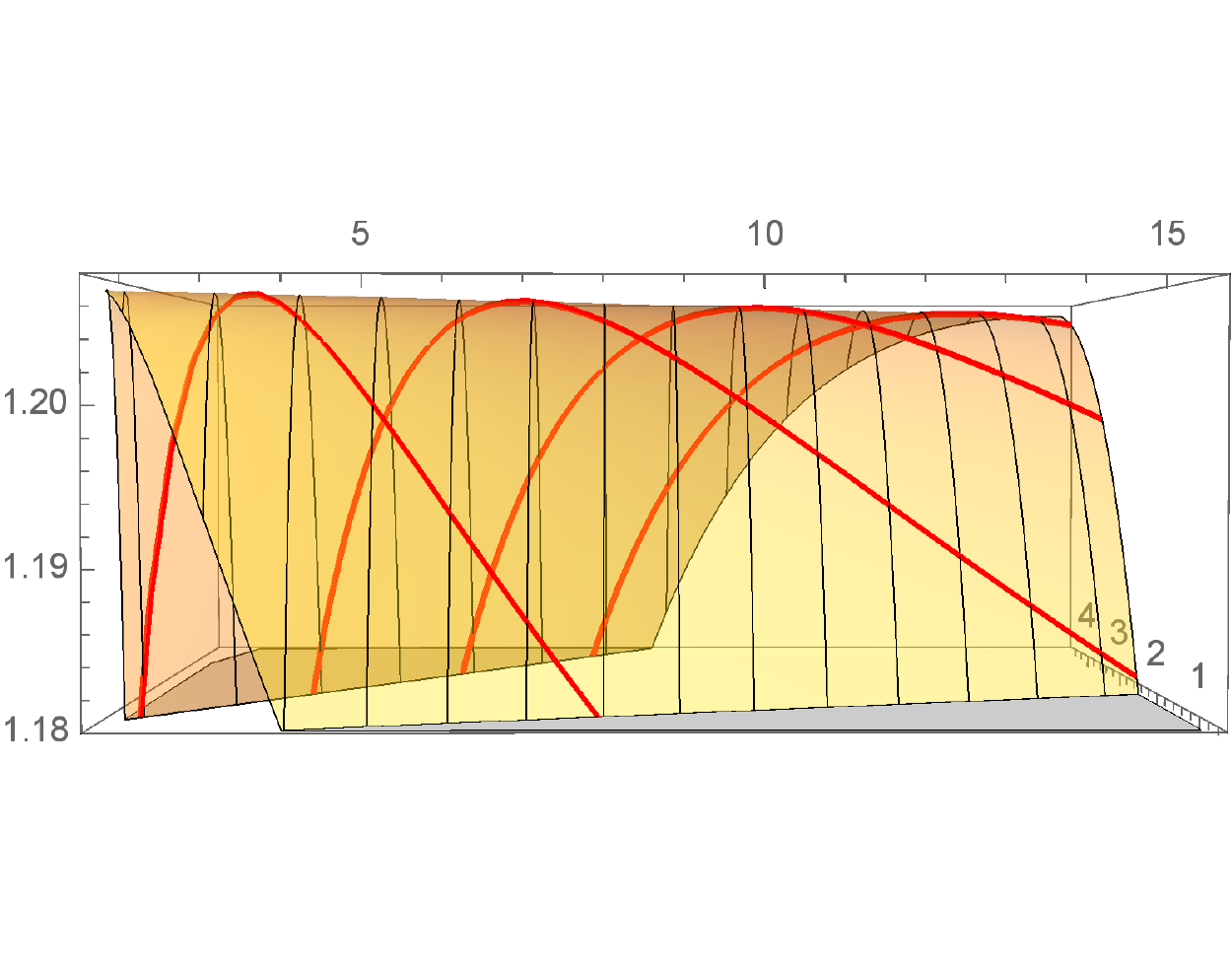}
 \caption{Graph of the function $(m,k) \mapsto \lambda_m''(k) = 1 + \frac{\sqrt{(2m-k)k}-k}{2m}$ showing the worst-case performance ratio with $m$ machines and $k$ long jobs} \label{graph of f}
\end{figure}
shows the function $(m,k) \mapsto \lambda_m''(k)$ with marked (thick) lines for integral values of~$k$. One can also see the concave dependence on $k$ (thin lines). 
\end{proof}

This concludes the proof of Theorem~\ref{theorem:Performance guarantee WSPT}. In Figure~\ref{graph of performance guarantee} the graph of
the function $m \mapsto \lambda_m(x_m,0,k_m)$, whose values at integral $m$ are exactly the worst-case approximation ratios for instances with $m$ machines, is depicted. The jumps and kinks occur when the number $k_m$ of long jobs in the worst-case instance changes. By taking the limit for $m \to \infty$, we obtain alternative proof of the performance ratio $\tfrac 1 2(1+\sqrt{2})$ by Kawaguchi and Kyan~\cite{KK86}, avoiding the somewhat complicated transformation and case distinction in the proof of Lemma~\ref{non-short jobs same length} and Schwiegelshohn's proof~\cite{Sch11}. For increasing $m$ the tight performance ratio converges quite quickly to $\tfrac 1 2(1+\sqrt 2)$: the difference lies in $O(1/m^2)$.
By plugging in the machine-dependent performance ratio into Theorem~\ref{Theorem Performance guarantee WSEPT}, we obtain the following performance ratio for the WSEPT rule.
\begin{corollary} \label{magenta performance guarantee}
 For instances of the problem $P|\bm p_j \sim \mathrm{stoch}|\E[\sum w_j \bm C_j]$ with $m$ machines the WSEPT rule has performance ratio
 \[1 + \frac 1 2 \cdot \frac{\sqrt{(2m-k_m)k_m}-k_m}{m} (1+\Delta).\]
\end{corollary}
This bound is better than the bound of Corollary~\ref{green performance guarantee} only if $m$ and $\Delta$ both are small. Even for two machines, it is outdone for large $\Delta$ (see Figure~\ref{fig:performance bound 2 machines}).
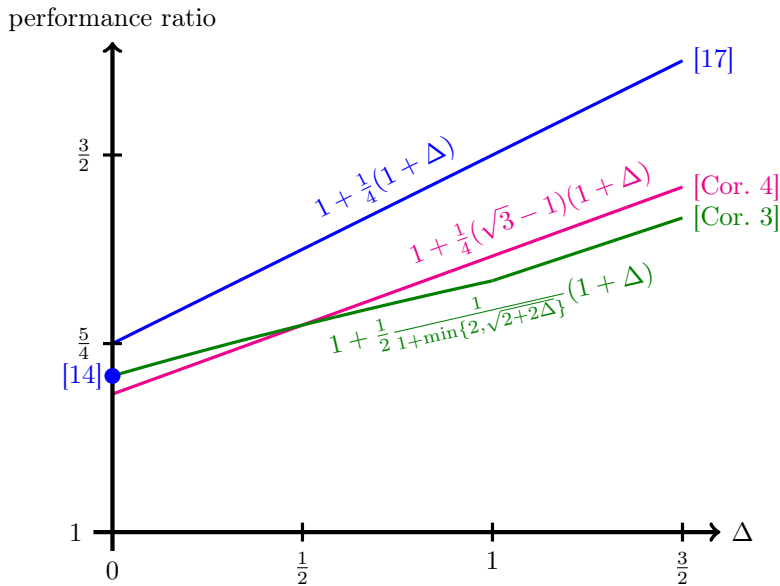
\begin{figure}
 \begin{tikzpicture}[x=5cm,y=10cm,domain=0:1]
  \draw [blue,very thick] (0,1.25) to node [sloped,above] {$1+\tfrac14(1+\Delta)$} +(1.5,1.5/4) node [right] {\cite{MSU99}};
  \draw [magenta,very thick] (0,{1+(sqrt(3)-1)/4}) to node [sloped, above,near end] {$1+\tfrac14(\sqrt3-1)(1+\Delta)$} +(1.5,{(sqrt(3)-1)/4*1.5}) node [right] {[Cor.~\ref{magenta performance guarantee}]};
  \draw [green!50!black,very thick] plot (\x,{1+1/2*1/(1+sqrt(2+2*\x))*(1+\x)});
  \node [rotate=13,green!50!black] at (1,1.29) {$1+\tfrac 1 2 \frac{1}{1+\min\{2,\sqrt{2 + 2\Delta}\}}(1+\Delta)$};
  \draw [green!50!black,very thick] (1,4/3) to +(0.5,0.5/6) node[right] {[Cor.~\ref{green performance guarantee}]};

  \draw [->,ultra thick] (-0.05,1) node [left] {$1$} -- (1.6,1) node [right] {$\Delta$}; 
  \draw [->,ultra thick] (0,0.975) node [below] {$0$} -- (0,1.65) node [above] {performance ratio};
  \foreach \i/\j in {0.5/\frac12,1/1,1.5/\frac32}
   \draw [very thick] (\i,0.9875) node [below] {$\j$} -- +(0,0.025);
  \foreach \i/\j in {1.25/\frac54,1.5/\frac32}
   \draw [very thick] (-0.025,\i) node [left] {$\j$} -- +(0.05,0);

  \draw [blue] (0,1.2071) node [left] {\cite{KK86}} node [circle,fill,minimum size=6pt,inner sep=0pt] {};
 \end{tikzpicture}
 \caption{Bound on WSEPT's performance ratio for two machines} \label{fig:performance bound 2 machines}
\end{figure}

\section{Open problem} \label{Open problem}

For every fixed value of $\Delta$, one obtains the machine-independent performance ratio of $1+\tfrac 1 2 (1+\min \{2,\sqrt{2+2\Delta}\})^{-1}(1+\Delta)$ for instances with squared coefficient of variation bounded by $\Delta$. On the other hand, for every fixed $m$ our performance bound tends to infinity when $\Delta$ goes to infinity, so that it does not imply a constant performance ratio (independent of $\Delta$) for instances with a constant number of machines. As far as the authors know, the question if such a $\Delta$-independent constant performance ratio of WSEPT for a fixed number of machines exists is still open. The examples of Cheung et al.~\cite{CFMM14} and of Im et al.~\cite{IMP15} only show that no constant performance ratio can be given when $\Delta$ and $m$ are allowed to go simultaneously to infinity.

\bibliography{references}

\end{document}